\documentclass[journal]{IEEEtran}
\usepackage[top=1.5cm, bottom=1.5cm, left=1.5cm, right=1.5cm]{geometry}
\usepackage{amsmath,amsfonts}
\usepackage{array}
\usepackage[caption=false,font=normalsize,labelfont=sf,textfont=sf]{subfig}
\usepackage{textcomp}
\usepackage{stfloats}
\usepackage{url}
\usepackage{verbatim}
\usepackage{graphicx}
\usepackage{cite}
\usepackage{flushend}
\usepackage{float}
\usepackage{hyperref}
\usepackage{titlesec}
\titlespacing*{\section} {0pt}{2ex}{2ex}
\usepackage{tipa}
\usepackage{textgreek}
\usepackage{amsmath, nccmath} 
\usepackage{mathtools}  
\newtheorem{theorem}{Theorem}
\usepackage[linesnumbered,ruled,vlined]{algorithm2e}

\SetCommentSty{mycommfont}
\SetKwInput{KwInput}{Input}                
\SetKwInput{KwOutput}{Output}              
\usepackage[table,xcdraw]{xcolor}
\usepackage{amssymb}

\begin{document}

\pagenumbering{gobble}
\title{Sparse Channel Estimation and Signal Recovery for Reduced-PAPR Visible Light Optical OFDM Systems Relying on Bayesian Learning}

\author{Shubham~Saxena,~\IEEEmembership{Graduate Student Member,~IEEE,} Suraj~Srivastava,~\IEEEmembership{Member,~IEEE,} Aditya~K.~Jagannatham,~\IEEEmembership{Senior Member,~IEEE}, and Lajos~Hanzo,~\IEEEmembership{Life Fellow, IEEE} \vspace{-5mm} 
\thanks{The work of Aditya K. Jagannatham was supported in part by the Qualcomm Innovation Fellowship; in part by the Qualcomm 6G UR Gift; in part by the Arun Kumar Chair Professorship. The work of S. Srivastava was supported in part by IIT Jodhpur’s Research Grant No. I/RIG/SUS/20240043; in part by Anusandhan National Research Foundation’s PM-ECRG/2024/478/ENS; and in part by Telecom Technology Development Fund (TTDF) under Grant TTDF/6G/368. S. Srivastava and A. K. Jagannatham jointly acknowledge the funding support provided by Anusandhan National Research Foundation's Advanced Research Grant ANRF/ARG/2025/005895/ENS. The work of Lajos Hanzo was supported by the Engineering and Physical Sciences Research Council (EPSRC) under projects: Platform for Driving Ultimate Connectivity (TITAN) (EP/X04047X/1; EP/Y037243/1); PerCom  EP/X01228X/1
, which is gratefully acknowledged. S. Srivastava, A. K. Jagannatham and Lajos Hanzo jointly acknowledge the funding support provided to ICON-project by DST and UKRI-EPSRC under India-UK Joint opportunity in Telecommunications Research.\\Copyright (c) 20xx IEEE. Personal use of this material is permitted. However, permission to use this material for any other purposes must be obtained from the IEEE by sending a request to pubs-permissions@ieee.org.\\
Shubham Saxena and Aditya K. Jagannatham are with 
the Department of Electrical Engineering, Indian Institute of Technology 
Kanpur, Kanpur-$208016$, India (e-mail: shubs20@iitk.ac.in; adityaj@iitk.ac.in). \\Suraj Srivastava is with the 
Department of Electrical Engineering, Indian Institute of Technology 
Jodhpur, Jodhpur, Rajasthan $342030$, India (email: 
surajsri@iitj.ac.in). \\ Lajos Hanzo is with the School of Electronics and 
Computer Science, University of Southampton, Southampton SO17 1BJ, U.K. 
(email: lh@ecs.soton.ac.uk)}}
\maketitle
\markboth{}{Saxena \MakeLowercase{\textit{et al.}}: Sparse Channel Estimation and Signal Recovery for Reduced-PAPR Visible Light Optical OFDM Systems Relying on Bayesian Learning}
\begin{abstract}
A sophisticated multipath channel impulse response (CIR) estimator is proposed, followed by sparse frequency-domain (FD) signal detection, which harnesses the simultaneous-sparsity innate in the multipath CIR and FD signals across measurement vectors, conceived for optical orthogonal frequency division multiplexing (O-OFDM) utilized for visible light communication (VLC) systems having reduced peak-to-average power ratio (PAPR). At the outset, we derive the input-output relationships for the asymmetrically clipped optical OFDM (ACO-OFDM) as well as for the direct current-biased optical OFDM (DCO-OFDM) systems. Next, traditional benchmarking methods are introduced for channel estimation (CE), including both the linear minimum mean square error (LMMSE) and least square (LS) methods, followed by an orthogonal matching pursuit (OMP)-based technique capable of exploiting the sparsity in the multipath CIR of the VLC system. Subsequently, an estimation technique based on a novel group-sparse OMP (GOMP) concept is proposed, which capitalizes on the group-sparsity in the delay domain of the CIR across measurement vectors, attributed to the multipath characteristics inherent in the non-line-of-sight (NLoS) components of the VLC channel model. Additionally, an advanced group-sparse multipath CIR estimation technique is put forth based on the group-sparse Bayesian learning (GBL) approach, which considerably reduces the pilot overhead. A low complexity version of GBL, termed LCGBL, is also developed that reduces the computational cost of GBL significantly. Consequently, the GOMP and GBL frameworks are also extended to the data detection of the sparse FD O-OFDM symbols, which utilizes the group-sparsity of the FD O-OFDM symbols across measurement vectors. The Bayesian Cramer-Rao lower bound (BCRLB) is computed to assess the estimation performance of the proposed CE techniques. Our simulations demonstrate that, despite its reduced pilot overhead, the proposed GBL technique outperforms the other conventional and sparse CE schemes in terms of its bit error-rate (BER), outage probability (OP), and normalized mean-square-error (NMSE), quantitatively reinforcing its superiority.
\end{abstract}
\section{Introduction}
\IEEEPARstart{T}{he} rapid expansion of high-speed mobile data traffic, fueled by the proliferation of the Internet of Things (IoT) and mobile devices like smartphones, poses a significant challenge in terms of radio frequency (RF) spectrum scarcity, culminating in a spectrum crunch. In order to quell this RF spectrum congestion, next-generation wireless systems may also harness visible light communication (VLC) \cite{babar2019near, sun2025joint}. Since rich spectral resources may be found in the VLC domain, the availability of low-cost devices, such as photodiodes (PDs) and light-emitting diodes (LEDs), also improves the appeal of VLC systems by providing simultaneous `communication' and `illumination'. To elaborate, the VLC band includes an enormous block of unregulated spectrum that is approximately $10~000$ times that of the RF band, along with low radiation, and robust resilience to electromagnetic interference. Additionally, VLC systems provide intrinsic security against eavesdropping, since the light is blocked by walls, thereby securing the transmitted data. Furthermore, VLC is deemed to be an energy and cost-efficient technique as it reuses the existing lighting resources for data transmission, and it is thus a popular `green' technology \cite{babar2019near,sun2025joint,wang2017learning}. Thus, due to the aforementioned factors, VLC may be viewed as a compelling complement to traditional RF communication. 

The modulation used in VLC is typically intensity modulation/direct detection (IM/DD), which relies on real-valued and non-negative signals \cite{wang2017learning,ghassemlooy2019optical}. In IM/DD modulation, at the transmitter side, the light intensity of the LED is modulated by the data to be transmitted, leading to the generation of an optical signal, while at the receiver side, the PD transforms the optical signal back to an electrical signal. The popular modulation techniques harnessed in VLC systems are pulse-width modulation, on-off keying, and $M$-ary pulse-amplitude modulation, which utilize a single carrier (SC). Nevertheless, at higher data rates, these techniques necessitate a sophisticated receiver relying on multitap equalization \cite{ghassemlooy2019optical,hei2019energy}, which leads to high complexity and potential performance degradation. 
Generally, the VLC channel has both line-of-sight (LoS) and non-LoS (NLoS) components. The NLoS paths are established when the light emitted from the LED source is reflected by various surfaces, such as walls and objects, before reaching the PD at the receiver. By contrast, the LoS path represents the shortest possible path to the PD receiver from the LED source, unimpeded by reflections or obstructions. Thus, the received signal exhibits time dispersion because of multipath propagation. The power delay characteristic of the VLC channel is examined in \cite{lee2011indoor}. In VLC systems, the ensuing delay spread causes intersymbol interference (ISI) \cite{219552}. Consequently, optical orthogonal frequency division multiplexing (O-OFDM) has garnered substantial interest from researchers because of its efficacy in offering a low-complexity single-tap FD equalization, enhanced ISI reduction, and improved resilience to multi-path propagation, in LED-based VLC systems \cite{popoola2014pilot}. 

Direct current-biased optical OFDM (DCO-OFDM) and asymmetrically clipped optical OFDM (ACO-OFDM) are the most dominant variants of O-OFDM modulation \cite{zhang2018multi}. Since VLC signals are naturally non-negative and real-valued, a DC bias is applied to the transmitted signal in the DCO-OFDM to guarantee unipolarity. Likewise, the ACO-OFDM achieves unipolarity by trimming the transmitted bipolar O-OFDM signals at zero, transmitting only the positive components. In comparison to DCO-OFDM, this characteristic makes ACO-OFDM more energy-efficient \cite{zhang2018multi}, while the former exhibits a superior spectral efficiency when compared to the latter. This advantage stems from the operational difference between ACO-OFDM and its DCO-OFDM counterpart: ACO-OFDM utilizes half of the odd subcarriers for data transmission, while DCO-OFDM employs only half of the subcarriers for transmitting the data symbols \cite{zhang2018multi}. Consequently, one must strike a balance between spectral efficiency and energy efficiency when choosing between the two O-OFDM techniques. More advanced O-OFDM techniques include pulse amplitude modulated discrete multitone OFDM \cite{ghassemlooy2019optical}, flip OFDM \cite{ghassemlooy2019optical}, asymmetrically clipped DCO-OFDM \cite{vappangi2018performance}, etc, which are combinations of ACO-OFDM modulation and DCO-OFDM modulation. However, using advanced O-OFDM techniques renders the transceiver design more challenging \cite{vappangi2018performance}. Therefore, due to the potential advantages of these O-OFDM approaches, both ACO-OFDM as well as DCO-OFDM principles are considered in this work. Although O-OFDM offers several benefits for VLC systems, it also presents significant challenges that require special attention. A major impediment of O-OFDM is its excessive peak-to-average power ratio (PAPR) resulting from the superimposition of multiple subcarriers \cite{popoola2014pilot}. This is exacerbated by the restricted dynamic range of the LED, which is the central cause of non-linearity in a VLC system \cite{vappangi2018performance}. Therefore, the best strategy to eliminate these non-linear imperfections is to linearize the non-linear region using pre-distortion methods or focus primarily on PAPR reduction methods. In addition, conventional PAPR reduction methods cannot be directly utilized in a VLC system due to the transmission of real and non-negative signals \cite{vappangi2018performance}. Another impediment in an O-OFDM-VLC system is the performance degradation arising due to interference from the background noise, natural light source, and the multipath VLC effect of the channel. Naturally, an O-OFDM-VLC receiver requires both equalization and channel estimation (CE). Thus, to address the challenges associated with high PAPR and CE in the VLC system, we propose an efficient CE method specifically designed for a reduced-PAPR O-OFDM-VLC system. The literature pertaining to this context is highlighted in the following subsection.
\begin{table*}
\label{tab:my-table}
\begin{center}
\caption{Boldly contrasting our contributions to the literature}
\resizebox{!}{!}{%
\begin{tabular}{|l|c|c|c|c|c|c|c|c|c|c|c|c|}
\hline
\textbf{Features}           &  \cite{vappangi2018performance}&\cite{ahmad2020papr}  & \cite{wang2022multi} & \cite{yang2021artificial}  & \cite{sharifi2020compressive}  & \cite{wang2020papr} & \cite{deng2018novel} & \cite{miao2020adaptive} & \cite{chen2016adaptive} & \cite{chen2017sparse} & \cite{shub} & \textbf{Proposed} \\ \hline
LoS channel model           & \checkmark & \checkmark & \checkmark & \checkmark &  & \checkmark & \checkmark & \checkmark & \checkmark & \checkmark  & \checkmark  & \checkmark               \\ \hline
DCO-OFDM                    & \checkmark & &  & \checkmark & \checkmark &  \checkmark & \checkmark & \checkmark & \checkmark & \checkmark & \checkmark    & \checkmark             \\ \hline
NLoS channel model          & \checkmark & \checkmark &  & \checkmark &  &  & & \checkmark & \checkmark & & \checkmark  & \checkmark               \\ \hline
ACO-OFDM                    & \checkmark & \checkmark  & \checkmark &  &  &  \checkmark &  &  & & & \checkmark   & \checkmark              \\ \hline
Sparse CIR                &  &  &  &  &  &  &  &  &  &  & \checkmark      & \checkmark           \\ \hline
\textbf{CS-based reduced-PAPR} &  &  &  &  & \checkmark &   &  &  &  &  & & \checkmark                 \\ \hline
\textbf{MMV} & &  &  &  & \checkmark &   &  &  &  & &  & \checkmark \\ \hline 
\textbf{MMV-based BCRLB}          &      &  &  &  &  &   &  &  &  &   &   & \checkmark               \\ \hline
\textbf{OMP-based CE/SR} &  &  &  &  & SR &   &  & CE &  &  & CE& \textbf{CE, SR}                 \\ \hline
\textbf{BL-based CE/SR} &  &  &  &  &  &  &  & CE & CE &  CE & CE & \textbf{CE, SR}    \\ \hline
\textbf{GOMP-based CE/SR}            &      &  &  &  &  &   &  &  &  &  & & \textbf{CE, SR}               \\ \hline
\textbf{LCGBL-based CE}            &      &  &  &  &  &   &  &  &  &  & & \textbf{CE}               \\ \hline
\textbf{GBL-based CE/SR}  & & &  &  &  &  &  &  & & & & \textbf{CE, SR}               \\ \hline
\end{tabular}%
}
\end{center}
\end{table*}
\subsection{Literature review}
Numerous solutions were offered in the existing body of research to overcome the PAPR problem of OFDM systems. These include signal distortion methods such as (i) companding (ii) clipping and filtering (iii) peak cancellation carrier (iv) peak windowing \cite{sandoval2017hybrid}. These methods reduce the PAPR by altering the O-OFDM signal prior to transmission. In the clipping and filtering-based techniques \cite{sandoval2017hybrid}, high peaks are clipped before transmission, which results in signal distortion. Companding approaches use non-linear transformation of the signal for mitigating the PAPR values \cite{sandoval2017hybrid}. This non-linear operation results in performance degradation as it compromises the orthogonality of the O-OFDM subcarriers. Thus, while the signal distortion techniques are straightforward to apply, they generate clipping distortion, which degrades the symbol error-rate (SER). Other approaches rely on multiple signaling and probabilistic techniques \cite{sandoval2017hybrid}, which produce different O-OFDM signals conveying identical information, with the signal characterized by the lowest PAPR being chosen for transmission. These methods include tone reservation, selected mapping, partial transmit sequence, and pilot-assisted solutions \cite{sandoval2017hybrid}. The limitation of these methods is that they require side information in addition to the data, which increases the computational complexity and reduces the data rate. Precoding is one of the alternative approaches for reducing the PAPR \cite{sandoval2017hybrid}. You and Kahn \cite{you2000average}, suggested a block coding strategy for minimizing the PAPR in an O-OFDM system; nevertheless, this method mandates additional bandwidth as well as imposes a complexity overhead. Zhou and Qiao \cite{zhou2015low1}, recommended employing the discrete Hartley transform to alleviate the PAPR of an ACO-OFDM system. 
Thus, these precoding-based approaches result in bit error-rate (BER) improvement and PAPR reduction, albeit at the cost of reduced data rate and additional memory requirements. 

Recently, compressive sensing (CS) based approaches have been widely used for PAPR reduction. Al-Safadi \textit{et al.} \cite{al2009reducing}, employed reserved frequency-domain (FD) tones and exploited CS for signal recovery (SR).  Ghassemlooy \textit{et al.} \cite{ghassemlooy2017papr}, use an orthogonal matching pursuit (OMP) algorithm for SR and perform PAPR reduction for ACO-OFDM signals by relying on a sophisticated Toeplitz matrix-based technique. Popoola \textit{et al.} \cite{popoola2014pilot} investigate the use of a pilot signal to reduce the PAPR in an O-OFDM intensity-modulated optical wireless communication system, employing a maximum likelihood estimator with an oversampling factor of $4$. Sharifi \textit{et al.} \cite{sharifi2020compressive} study a CS and Zadoff-Chu (ZC) sequence-based method for compressing time-domain (TD) DCO-OFDM signals at the transmitter, using the multiple measurement vector (MMV) approach and OMP algorithm for SR. Their simulations, performed with quadrature amplitude modulation
(QAM) modulation schemes and various compression factors, show that increasing the compression factor reduces PAPR. \textcolor{black}{Chen \textit{et al.} \cite{chen2020comparison} examined several widely used precoding techniques for PAPR reduction and reported that ZC precoding offers superior PAPR performance compared with competing alternatives. Similarly, Ma \textit{et al.} \cite{ma2019performance} investigated a constant amplitude zero autocorrelation sequence (CAZAC) precoding framework derived from the ZC structure and demonstrated strong PAPR reduction capability. Hu \textit{et al.} \cite{hu2012peak} further showed that pilot sequences derived from subsampled ZC sequences can improve PAPR performance in pilot-aided OFDM systems. In the VLC context, Guo \textit{et al.} \cite{guo2020experimental,guo2018experimental} reported that ZC precoding outperforms orthogonal circulant matrix transform (OCT), Walsh-Hadamard transform (WHT), discrete-Hartley transform (DHT), and discrete cosine transform (DCT)-based precoding approaches for PAPR reduction. Moreover, Baig \textit{et al.} \cite{baig2011zcmt,baig2010new,baig2010papr,baig2013zcmt} demonstrated the superiority of ZC-matrix-transform-based precoding over conventional transform precoders such as WHT, DHT, and DCT for reducing PAPR.} Azarnia \textit{et al.} \cite{azarnia2022performance} introduced a CS-based method to reduce PAPR in OFDM systems. Their scheme employs a mapping matrix for signal compression and the group least absolute shrinkage and selection operator (G-LASSO) algorithm for recovery, with performance evaluated through PAPR and BER metrics at an oversampling factor of $4$. Compared to conventional methods, the CS approach offers ease of implementation, independence from side information, and a low computational complexity \cite{sharifi2020compressive}. Consequently, we also conceive a CS-based method for PAPR reduction.

For CE and equalization in the O-OFDM-VLC systems, Zhao \textit{et al.} \cite{zhao2013channel} employ established techniques, specifically the linear minimum mean square error (LMMSE) and least square (LS) methods for a VLC channel impulse response (CIR) of $L_h = 6$ taps. While the LS method is straightforward to implement, it exhibits sensitivity to noise \cite{du2016channel}. Compared to the LS, the LMMSE algorithm provides improved channel state information (CSI), although it is slightly more complex \cite{du2016channel}. For high-rate systems, Zhou \textit{et al.} \cite{zhou2014impact} investigated higher-order reflections of the multipath VLC systems. The study suggests that traditional VLC channel models, which solely account for initial reflections, may not be accurate enough for high-rate VLC systems. Gong \textit{et al.} \cite{gong2015channel}, use maximum likelihood sequence detection and LMMSE for ISI removal and develop an LS-based CE scheme, modeling the VLC channel as a linear time-invariant (LTI) system and Poisson distributed with only NLoS components. 
\textcolor{black}{The authors of \cite{lee2019deep} employed a deep neural network (DNN)-based framework for
dimmable VLC systems, where the encoder and decoder pair is replaced by a DNN-based VLC transceiver. However, their simulations were limited to on-off-keying
modulation, and the approach generally demands large training data and lacks interpretability.} Zhang \textit{et al.} in \cite{zhang2016enhancing} consider the ACO-OFDM and use OMP as well as LS together with the discrete Fourier transform (DFT) for performing CE in the VLC system. In reality, a typical VLC channel comprises a LoS and several NLoS components, which make this system more complex \cite{chen2016adaptive,shub, carruthers2002iterative, barry1993simulation}. The multipath characteristics of the NLoS components result in sparsity in the VLC channel within the delay-domain \cite{shub,zhao2014compressed,du2016channel,shi2020adaptive,zhang2016enhancing}. As for CE in a sparse VLC system, Du \textit{et al.} \cite{du2016channel} developed a CS-based approach that relies on a dynamic over-complete dictionary matrix. The key shortcoming of the approach discussed in \cite{du2016channel} is that it has a high computational complexity arising due to having an over-complete dictionary matrix, which leads to its slow estimation speed. In their seminal paper Shi \textit{et al.} \cite{shi2020adaptive}, proposed a sparse LS method for CE. The key assumption in \cite{shi2020adaptive}, which forms the foundation of their approach, concerns modeling the VLC CIR components as zero-mean Gaussian random variables. Needless to say, this is not guaranteed in a practical VLC system, which can potentially lead to poor performance. Chen \textit{et al.} \cite{chen2017sparse} presented a sparse Bayesian learning (BL) approach relying on a relevance vector machine for CE in O-OFDM-VLC systems. Their method focuses solely on the LoS channel gain and utilizes complex-valued training symbols. However, it is essential to note that their findings are exclusively evaluated for DCO-OFDM systems and assume the imaginary and real components of the transmitted pilot symbols to be identical. This limitation constrains the broader generality of their research. Compared to the standard LS and LMMSE techniques, CS-based CE approaches substantially enhance the estimation accuracy \cite{zhao2014compressed}. Meanwhile, the CS-based techniques presented in the VLC literature require insights into the sparsity characteristics of the channel and rely on the choice of the stopping parameter and dictionary/measurement matrix, with modest deviations causing serious convergence errors and performance loss \cite{srivastava2021bayesian}. 

Our previous work \cite{shub}, introduced a sparse CE technique based on BL for DCO-OFDM and ACO-OFDM VLC systems. However, that framework was limited to a single measurement vector (SMV) model. Several studies have shown that the MMV model leads to performance improvement over the SMV model, particularly when the sparsity structures of the resultant vectors are similar. The MMV method offers the benefit of achieving sparser solutions compared to the SMV approach \cite{srivastava2021bayesian,srivastava2021sparse}. Furthermore, the analysis pertaining to PAPR reduction in the O-OFDM system is not included in \cite{shub}. Thus, to address the limitations of the existing CE techniques therein and also to address the inevitable problem of high PAPR in an O-OFDM system, we conceive a novel Group sparse Bayesian Learning (GBL)-aided simultaneous group sparse multipath CIR estimation scheme. This approach is also extended toward SR of the reduced-PAPR FD O-OFDM symbols in the VLC system. Our results demonstrate that the GBL-based technique conceived is highly effective in contrast to the existing CE and SR strategies, such as FOCal Underdetermined System Solver (FOCUSS) and OMP. Table \ref{tab:my-table} contrasts our contributions to the literature. Below is a summary of the main contributions of this work.
\subsection{Contributions} 
\begin{enumerate}

    \item A CS framework is developed to effectively exploit the inherent sparsity of the CIR in a reduced-PAPR wideband VLC system, while accounting for specular and diffusive reflections as well as the characteristics of the optical receiver, source, and reflectors. Within this framework, an MMV-based sparse multipath CIR model and a sparse over-sampled reduced-PAPR FD OFDM symbol model are formulated to achieve improved estimation and recovery performance compared to the conventional SMV-based model.

    \item The intrinsic sparsity of the delay-domain CIR in VLC systems is further utilized along with the group-sparsity of the multipath CIR across MMV. Based on this structure, an OMP-based group-sparse CIR recovery technique, referred to as GOMP, is proposed for CE in O-OFDM-VLC systems. This approach is specifically designed to exploit the common sparse support shared across measurement vectors, thereby improving the reliability of CIR estimation. 

    \item Furthermore, a novel BL-based framework is introduced for multipath CIR estimation in O-OFDM-VLC systems by exploiting the group-sparsity of the multipath CIR across measurement vectors. Unlike conventional CE schemes, the proposed GBL method estimates the multipath CIR vectors using only a limited number of pilot subcarriers. Consequently, it substantially reduces pilot overhead while improving spectral efficiency. Another important advantage of the GBL method is that it does not require regularization or tuning parameters and demonstrates stable convergence behavior.
    \item To further improve computational efficiency, a low-complexity variant of GBL, termed LCGBL, is also developed. The LCGBL method significantly reduces the computational burden of GBL by replacing the inversion of a large-dimensional matrix with the inversion of two smaller matrices. Since the complexity of matrix inversion scales cubically with its dimension, this modification leads to a considerable reduction in overall computational complexity.
    \item Subsequently, the proposed GBL and GOMP techniques are extended to joint SR of group-sparse over-sampled FD O-OFDM symbols across measurement vectors, using the estimated group-sparse multipath CIR obtained from their respective CE methods. This extension improves BER performance compared to conventional non-sparse and sparse recovery techniques. 
    
    
    \item To establish a benchmark for the CE performance of the GBL-based technique, the Bayesian Cramér-Rao lower bound (BCRLB) is computed for the reduced-PAPR O-OFDM VLC system.
    \item The presented GBL-based CE and SR techniques are evaluated for the two most prominent O-OFDM systems, namely DCO-OFDM followed by ACO-OFDM, using multiple metrics, such as outage probability (OP), BER, pilot overhead ($\rho$), normalized mean-squared-error (NMSE), and NMSE variation with different numbers of frames. 
    
\end{enumerate}

\textit{Notations}: These notations are used throughout the paper: The operators $(\cdot)^\dagger$ denotes the pseudoinverse of a matrix. Furthermore, $\mathbb{R}^{M \times N}_{+}$ denotes the set of $M \times N$ matrices whose elements are non-negative real values, respectively. $\mathbb{E} \{ \cdot \}$ represents the statistical expectation operator, diag$\{ \cdot \}$ represents a diagonal matrix, Tr($\mathbf{A}$) denotes
the trace of the matrix $\mathbf{A}$, and arg~$f(.)$ returns the argument of the function $f(.)$. The function $\operatorname{Re}(.)$ denotes the real part and $\det(\cdot)$ is the determinant of the corresponding matrix. The quantity $\tilde{(\cdot)}$ denotes a variable in FD, $\overline{(\cdot)}$ is a vector in FD, and $\widehat{(\cdot)}$ denotes an estimate of the variable.

\section{Multipath CS-based reduced-PAPR O-OFDM-VLC System Models}
The two most often used O-OFDM systems, DCO-OFDM and ACO-OFDM, are detailed in this section. Unlike RF-based OFDM systems, DCO-OFDM and ACO-OFDM systems directly modulate the illumination of the LEDs. Hence, the output signal must be real-valued and non-negative. 

\subsection{DCO-OFDM system model }
Fig. \ref{DCO1_MMV} depicts the schematic of the MMV-based DCO-OFDM transceiver. Initially, the serial bit sequence of a single OFDM block having $M$ subcarriers is converted to complex-valued symbols via QAM. The resultant modulated signal is hosted by: $\mathbf{\overline{x}} = [\tilde{x}_0, \hdots, \tilde{x}_{M-2}, \tilde{x}_{M-1}]^T\in \mathbb{C}^{M \times 1}$, so that the information symbols are from $\tilde{x}_1 ~\text{to}~ \tilde{x}_{M/2 - 1}$ and follow the characteristic, $\tilde{x}_j = \tilde{x}^{\ast}_{M-j} ~ \text{where} ~ M/2 +1 \leq j \leq M-1$. Moreover, to eliminate the DC signal, we set $\tilde{x}_0 = \tilde{x}_{M/2} = 0$. The real-valued TD signal of the vector $\mathbf{\overline{x}}$ is determined using the $M$-point inverse fast Fourier transform (IFFT), given as $\mathbf{x} = \mathbf{W}\mathbf{\overline{x}} $, where $\mathbf{W} = \{ w_{n,k} \}^{M-1}_{n,k = 0} \in \mathbb{C}^{M \times M} $ and $w_{n,k} = \frac{1}{M}e^{j \left(\frac{2\pi n k}{M} \right)}$. 
\textcolor{black}{The resultant signal undergoes parallel-to-serial (P/S) conversion, followed by the concatenation of a cyclic prefix (CP), assuming perfect synchronization at the receiver.} With the objective of eliminating the ISI, the length of the CP ($L_{\text{CP}}$) is kept significantly higher than the VLC multipath channel's delay spread. The waveform $x(t)$ is then produced by the digital-to-analog converter (DAC) and subsequently applying a low-pass filter (LPF). For intensity modulation, $x(t)$ must be unipolar in nature, therefore, a DC bias ($B_{\text{DC}}$) is added to $x(t)$, which is given by $B_{\text{DC}}=  \text{\textscriptv} \sqrt{\mathbb{E} \{ x^2(t)\} }$, and the constant \textscriptv ~is chosen to satisfy $B_{\text{DC}} = 10\log(\text{\textscriptv}^2+1)$ \cite{kahn1997wireless}. The resultant unipolar signal is described as $x_B(t) = x(t) + B_{\text{DC}}$. Consequently, the remaining negative signal contribution is forced to zero, due to the introduction of $B_{\text{DC}}$. The signal obtained is subsequently converted into an optical signal, followed by transmitting it through the multipath VLC channel described by the $L_h$-length vector $\mathbf{h} = \left[ h(0), h(1), h(2), \hdots,h(L_h-1)\right]^T \in \mathbb{R}^{L_h \times 1}_{+}$. The receiver side processing includes CP elimination, converting serial data to parallel (S/P), and the fast Fourier transform (FFT). Thus, in the VLC system, the input-output architecture for the $k^{th}$ subcarrier is given as follows:
\begin{figure}[t]
\centering
\includegraphics[width=\linewidth,height=55mm]{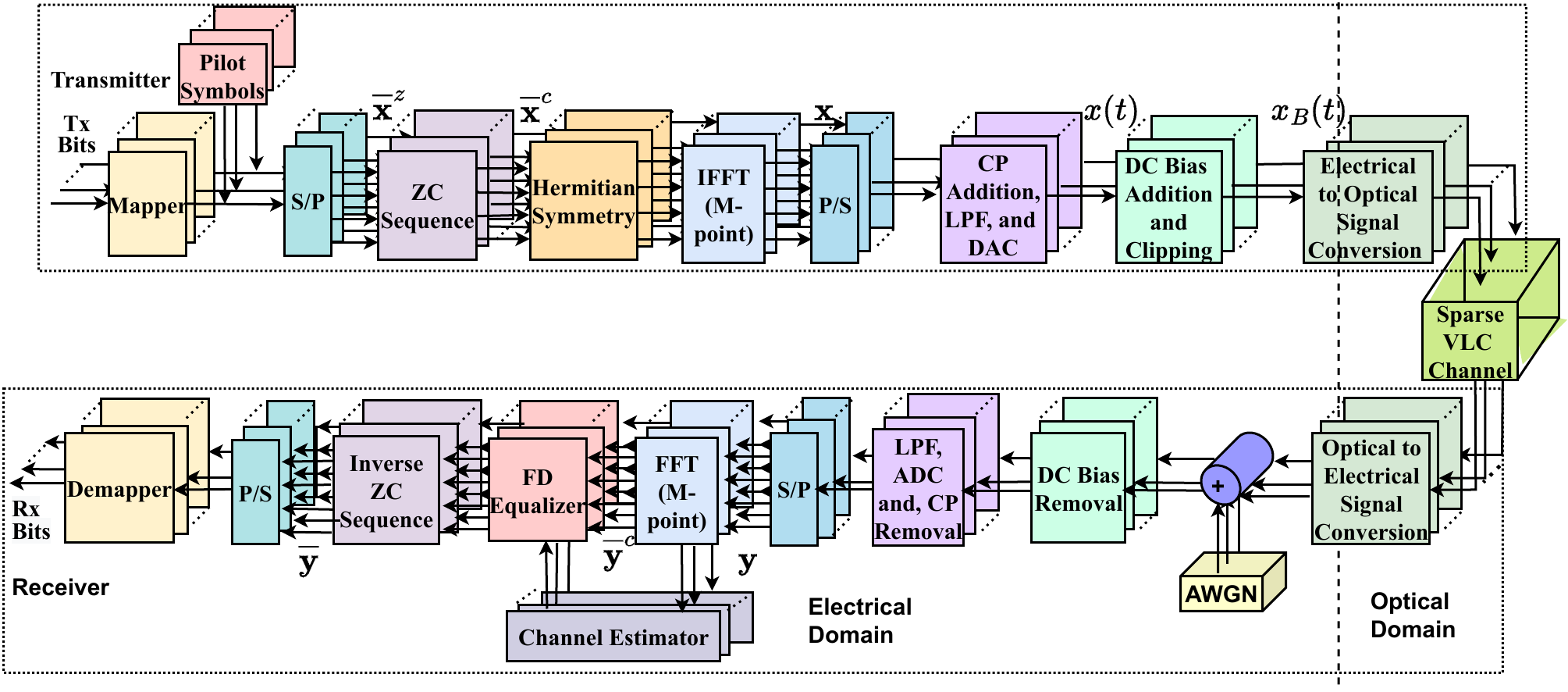}
\caption{Schematic diagram of an MMV-based reduced-PAPR DCO-OFDM VLC transmitter and receiver.}
\label{DCO1_MMV}
\end{figure}
\begin{equation}\label{eq5}
\tilde{y}_k = \tilde{h}_k \tilde{x}_k + \tilde{\textit{n}}_k, ~~k = 0,1,\hdots,M-1,
\end{equation}
where $\tilde{h}_k$ denotes the channel's transfer function (CTF), defined as $\tilde{h}_k = \sum_{l=0}^{L_h-1} h(l)e^{\frac{-j2\pi nk}{M}}$, while $\tilde{n}_k$ is the complex additive white Gaussian noise (AWGN). The vector representation of the received signal after the FFT output is $\mathbf{\overline{y}}=[\tilde{y}_0, \tilde{y}_1, \hdots, \tilde{y}_{M-2}, \tilde{y}_{M-1}]^T \in \mathbb{C}^{M\times 1}$, which is modeled as
\begin{equation}\label{DCO_out}
    \mathbf{\overline{y}} = \text{diag}\{ \mathbf{\overline{h}}\}\mathbf{\overline{x}} + \boldsymbol{\overline{n}},
\end{equation}
where we have $\boldsymbol{\overline{n}} = [\tilde{\mathit{n}}_0, \tilde{\mathit{n}}_1, \hdots, \tilde{\mathit{n}}_{M-1}]^T \in \mathbb{C}^{M \times 1}$, and the VLC channel vector is given by $\mathbf{\overline{h}} = [\tilde{h}_0, \tilde{h}_1, \hdots, \tilde{h}_{M-1}]^T \in \mathbb{C}^{M \times 1}$. 
As discussed in (\ref{DCO_out}), the DCO-OFDM system model comprises $M$-parallel flat-fading systems. Hence, to recover the FD signal $\widehat{x}_k$ corresponding to the $k^{th}$ subcarrier, only single-tap equalization is needed at the receiver side, i.e., we have $\widehat{x}_k = \displaystyle \frac{\tilde{y}_k}{\tilde{h}_k}$. The desirable transmitted bits are then obtained by the demapper/demodulator, which maps the complex symbols $\widehat{x}_k$ to the corresponding bits.
\subsection{ACO-OFDM system model}
Fig. \ref{ACO1_MMV} illustrates the MMV-based ACO-OFDM transmitter and receiver. The differentiating feature of this O-OFDM with respect to DCO-OFDM lies in the arrangement of the data symbols. Explicitly, in the ACO-OFDM system, in line with the Hermitian symmetry of $\tilde{x}_j = \tilde{x}^{\ast}_{M-j}, ~ \text{where} ~ M/2 +1 \leq j \leq M-1$, the data symbols are exclusively allocated to the odd subcarriers among the initial $M/2$ subcarriers. Thus, the effective number of subcarriers that carry information is $M/4$, as compared to $M/2$ in DCO-OFDM. 
Following the IFFT, the resultant signal in the TD is real-valued and follows the anti-symmetry property \cite{shub}.
\begin{figure}
\centering
\includegraphics[width=\linewidth,height=55mm]{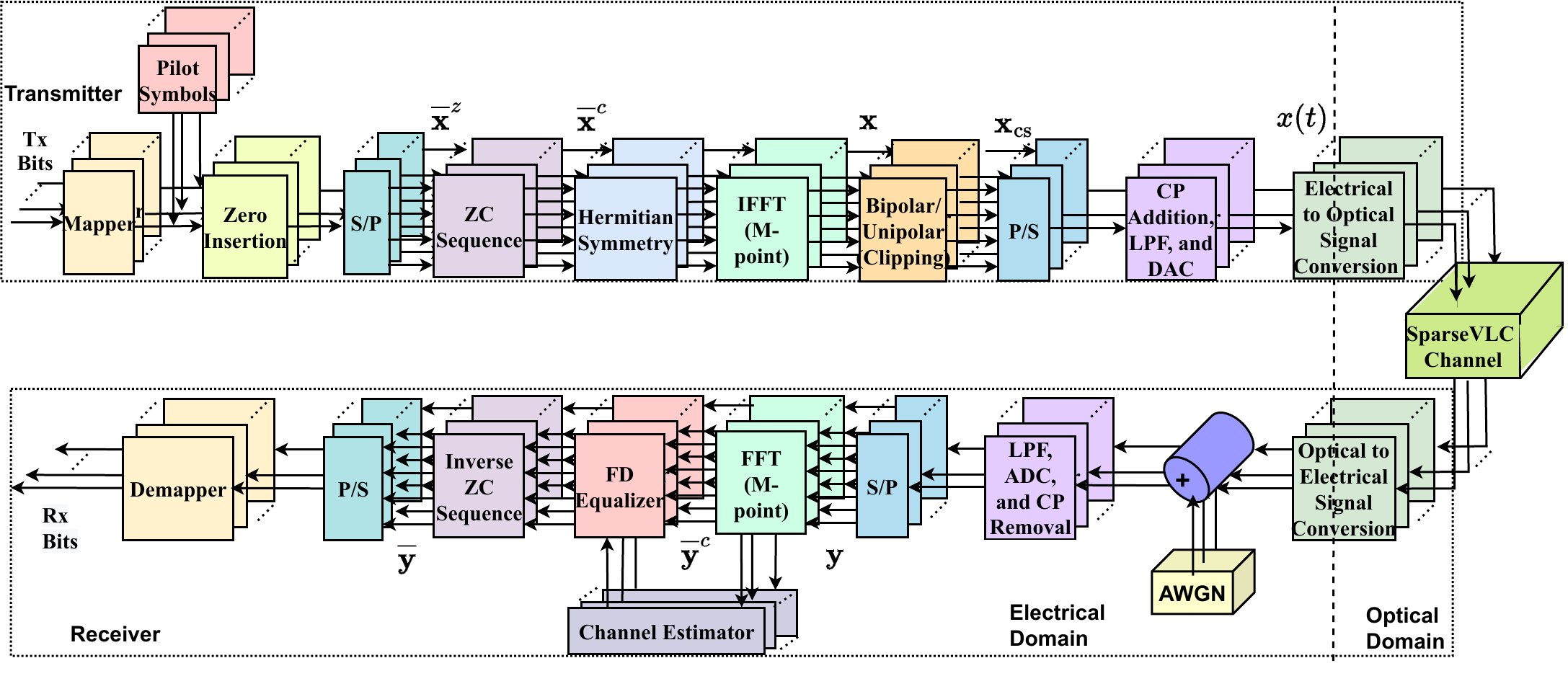}
\caption{Schematic diagram of an MMV-based reduced-PAPR ACO-OFDM VLC transmitter and receiver.}
\label{ACO1_MMV}
\end{figure}
Thus, the first halves of the samples possess the opposite signs, but identical amplitudes to those of the second halves of the samples. Hence, the samples with negative values are clipped to zero without the need for a DC bias. The ensuing signal is then sent successively through the P/S, CP addition, DAC, and LPF modules of Fig. \ref{ACO1_MMV}, producing $x(t)$. Subsequently, the resultant signal is input into the optical modulator and transmitted through the frequency-selective multipath channel within the VLC system \cite{dissanayake2013comparison}. 
The primary distinction in receiver-side processing between the ACO-OFDM and DCO-OFDM lies in the post-FFT extraction of odd subcarriers for the subsequent equalization and demapping processes. The following section describes the CS-based reduced-PAPR model of the
O-OFDM VLC system.
\subsection{CS-based reduced-PAPR O-OFDM-VLC system model}
This section describes the CS-based reduced-PAPR model of the O-OFDM VLC system. Let the FD modulated signal of O-OFDM be denoted as $\mathbf{\overline{x}} \in \mathbb{C}^{N \times 1}$, which is zero-padded, while using the over-sampling factor $l_z$, where $N_z = Nl_z$ as seen in Fig. \ref{a}\subref{1n}. Over-sampling is required to capture all signal peaks accurately, with an over-sampling factor of $l_z=4$ chosen to approximate the continuous TD OFDM signal \cite{popoola2014pilot}. Thus, the number of zero padded symbols is $(l_z-1)N$, and the over-sampled FD O-OFDM symbol is $\mathbf{\overline{x}}^z \in \mathbb{C}^{N_z \times 1}$ (Referred as SMV in the Fig. \ref{a}\subref{1n}). Due to over-sampling, most of the subcarriers of the signal $\mathbf{\overline{x}}^z$ are zero (shown with non-colored boxes). Therefore, the over-sampled FD O-OFDM symbols exhibit sparsity \cite{sharifi2020compressive}, rendering this a CS-based sparse SR problem. Next, to reduce the PAPR of the O-OFDM symbols, the over-sampled FD symbols are mapped using the precoding matrix $\mathbf{\Phi} \in \mathbb{R}^{M \times N_z}$, which is constructed using a real-valued ZC sequence as follows \cite{sharifi2020compressive}
\begin{equation}
\mathbf{\Phi} = 
\begin{bmatrix}
\phi_{0,0} & \phi_{0,1} & \cdots & \phi_{0,N_z-1} \\
\phi_{1,0} & \phi_{1,1} & \cdots &\phi_{1,N_z-1} \\
\vdots  & \vdots  & \ddots & \vdots  \\
\phi_{M-1,0} & \phi_{M-1,1} & \cdots & \phi_{M-1,N_z-1} 
\end{bmatrix},
\end{equation}
where $ \phi_{u,v} = \frac{1}{\sqrt{N_z}}\cos\left[ \frac{\pi}{MN_z} (u + vM)^2 \right]$ and $M = \alpha N_z$. \textcolor{black}{The main motivation for employing ZC precoding is that it provides a highly desirable combination of properties for both PAPR reduction and signal recovery. In particular, ZC sequences are CAZAC, characterized by constant modulus, ideal periodic autocorrelation, and low cross-correlation. These features make them especially attractive for O-OFDM systems, particularly when peak reduction is required without introducing nonlinear distortion or excessive signaling overhead. Furthermore, ZC sequences preserve their structural characteristics under FFT/IFFT operations, which further supports their suitability for transform-domain OFDM processing \cite{chen2020comparison,ma2019performance,hu2012peak,guo2020experimental,guo2018experimental,baig2011zcmt,baig2010new,baig2010papr,baig2013zcmt}.} Thus, the compressed signal is expressed as $\mathbf{\overline{x}}^c = \mathbf{\Phi} \mathbf{\overline{x}}^z$, where $\mathbf{\overline{x}}^c =  [\tilde{x}_0^c, \tilde{x}_1^c, \hdots, \tilde{x}_{M-1}^c]^T \in \mathbb{C}^{M \times 1}$. The resultant signal obeys Hermitian symmetry. Accordingly, the subcarriers $\tilde{x}_1^c ~\text{to}~ \tilde{x}_{M/2 - 1}^c$ satisfy the characteristic $\tilde{x}_j^c = (\tilde{x}^{c}_{M-j})^{\ast},~\text{for} ~ M/2 +1 \leq j \leq M-1$, and we set $\tilde{x}_0^c = \tilde{x}_{M/2}^c = 0$ to eliminate the DC signal. 
\begin{figure*}[h]
	\centering
        \captionsetup[subfigure]{justification=centering}
	\subfloat[]{\label{1n}\includegraphics[width=0.5\linewidth,height=65mm]{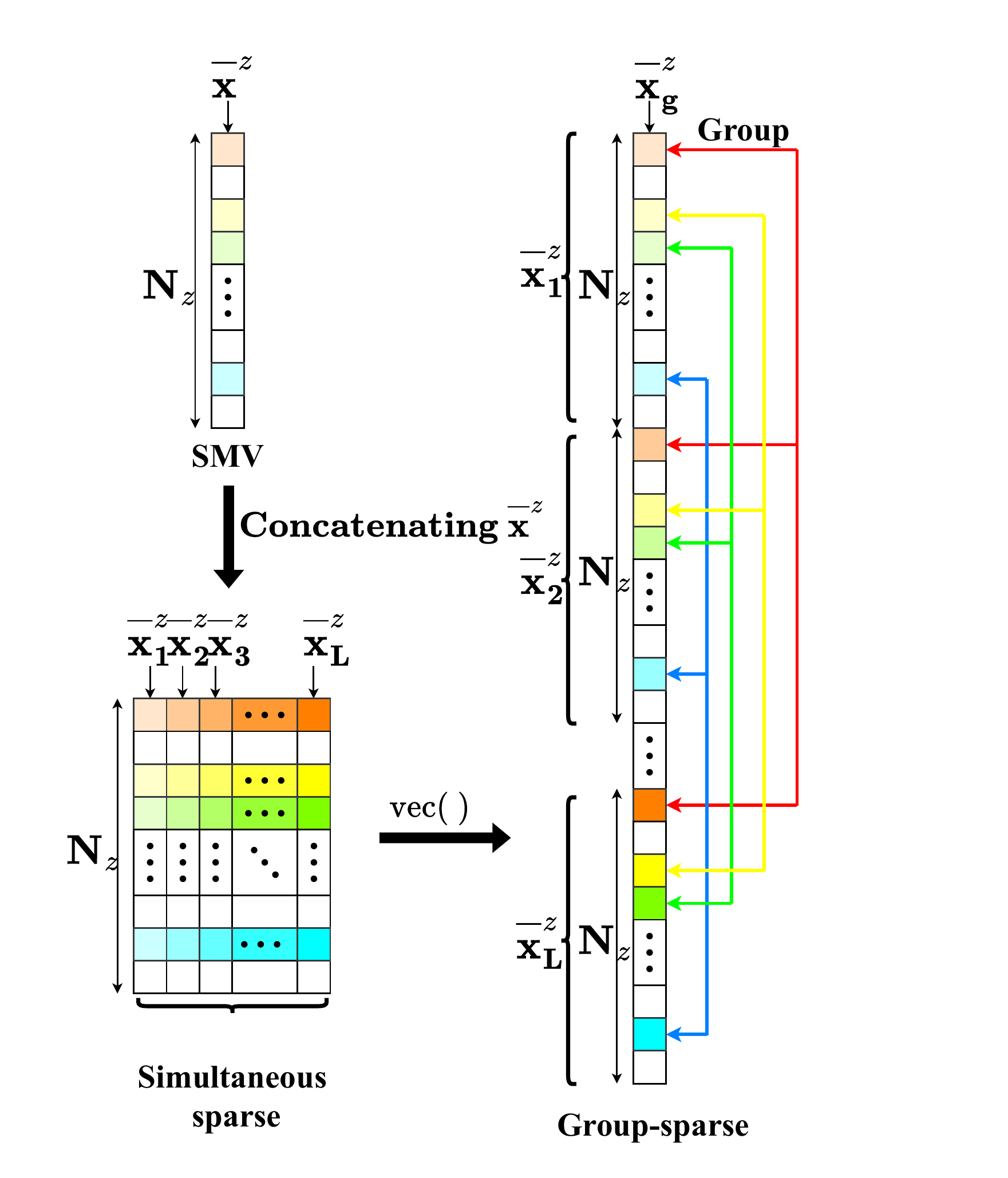}}\hspace{-10mm}
	\subfloat[]{\label{2}\includegraphics[width=0.48\linewidth,height=65mm]{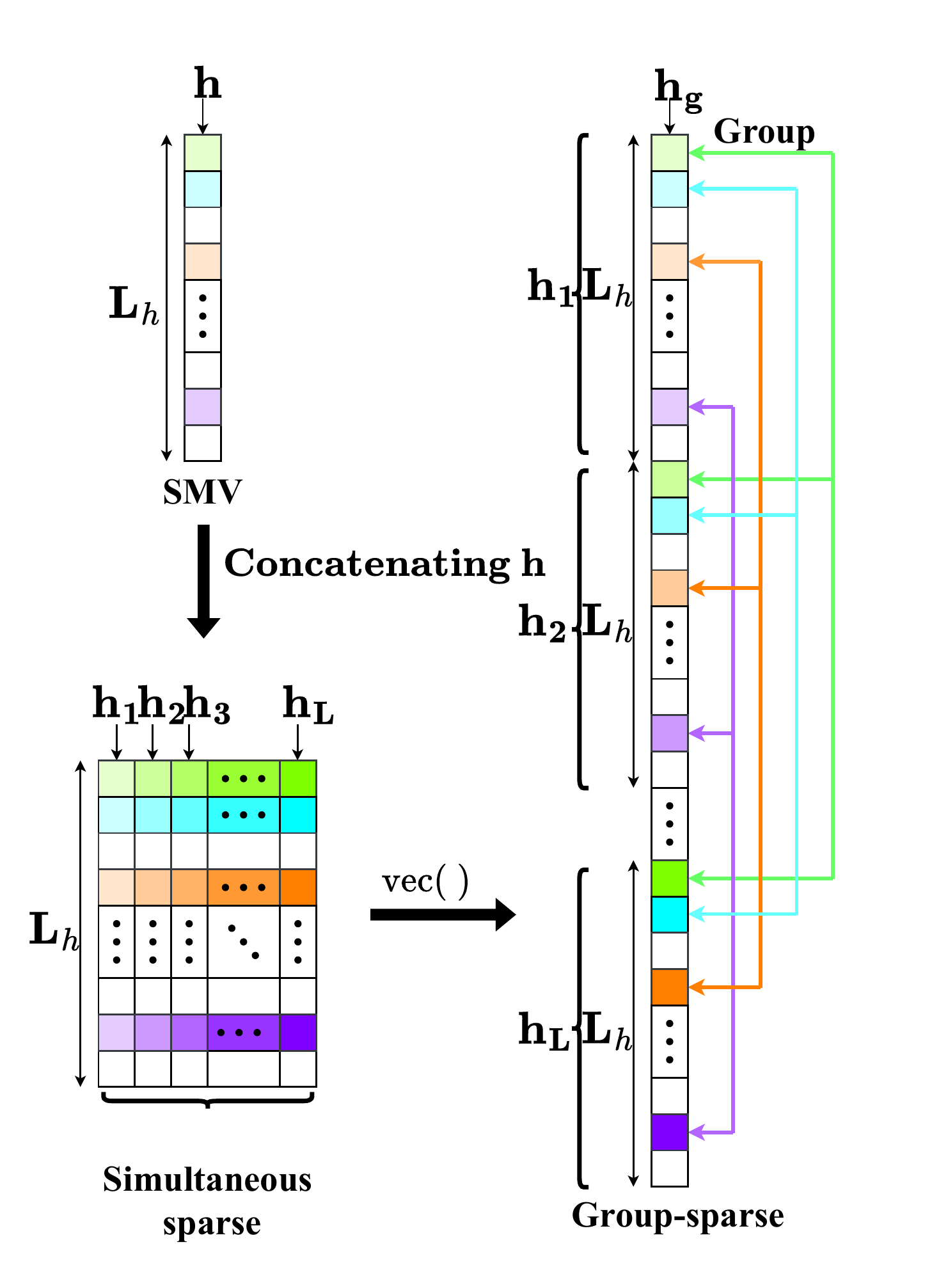}}
	\caption{\textcolor{black}{Simultaneous-sparse and group-sparse structure of the (a) over-sampled FD O-OFDM symbol vector $\mathbf{\overline{x}}^z_{\mathbf{g}}$; (b) multipath VLC CIR vector $\mathbf{h_g}$. Here, the colored boxes represent non-zero entries, with various colors indicating different SMVs having the same sparsity profile, while the non-colored boxes denote zero entries.}}
	\label{a}
\end{figure*}
The vector $\mathbf{\overline{x}}^c$ is subjected to the $M$-point IFFT to generate a TD signal, which is real-valued and it is given as $\mathbf{x} = \mathbf{W} \mathbf{\overline{x}}^c $, where $\mathbf{W} = \{ w_{n,k} \}^{M-1}_{n,k = 0} \in \mathbb{C}^{M \times M} $ and $w_{n,k} = \frac{1}{M}e^{j \left(\frac{2\pi n k}{M} \right)}$. Following CP elimination, S/P conversion, and FFT at the receiver side, the FD output vector $\mathbf{\overline{y}}^c=[\tilde{y}_0^c, \tilde{y}_1^c, \hdots, \tilde{y}_{M-1}^c]^T \in \mathbb{C}^{M \times 1}$ of a reduced-PAPR VLC system is given by
\begin{equation}\label{MMV1}
    \mathbf{\overline{y}}^c =\text{diag}\{ \mathbf{\overline{h}}\}\mathbf{\overline{x}}^c + \mathbf{\overline{n}}^c = \text{diag}\{ \mathbf{\overline{h}}\}\mathbf{\Phi\overline{x}}^z + \mathbf{\overline{n}}^c,
\end{equation}
where $\mathbf{\overline{n}}^c \in \mathbb{C}^{M \times 1}$ is the complex AWGN vector. The resultant signal is passed to the inverse ZC sequence ($\mathbf{\Phi}^H$) to obtain the FD output vector $\mathbf{\overline{y}} \in \mathbb{C}^{N_z \times 1}$ given by
\begin{equation}\label{MMX}
    \mathbf{\overline{y}} = \mathbf{\Phi}^H \mathbf{\overline{y}}^c =\mathbf{\Phi}^H \text{diag}\{ \mathbf{\overline{h}}\}\mathbf{\Phi \overline{x}}^z + \mathbf{\Phi}^H\mathbf{\overline{n}}^c = \mathbf{A}\mathbf{\overline{x}}^z + \mathbf{\overline{n}},
\end{equation}
where $\mathbf{A} = \mathbf{\Phi}^H \text{diag}\{ \mathbf{\overline{h}}\}\mathbf{\Phi} \in \mathbb{C}^{N_z \times N_z}$ is the measurement matrix, and $\mathbf{\overline{n}} = \mathbf{\Phi}^H\mathbf{\overline{n}}^c = [\tilde{n}_0, \tilde{n}_1, \hdots, \tilde{n}_{N_z-1}]^T$. Here, $\tilde{n}_k$ is the complex AWGN with variance $\sigma_{\text{AWGN}}^2$ and zero mean. 
\textcolor{black}{In the considered quasi-static indoor VLC scenario, the MMV model is formed over a short observation interval during which the transmitter, receiver, and propagation environment remain effectively unchanged. Hence, the sparsity pattern, namely the locations of the significant entries in the over-sampled FD O-OFDM symbol vectors, is assumed to remain identical across the $L$ measurement snapshots, as shown in Fig. \ref{a}\subref{1n}. Under this assumption, the snapshots exhibit a simultaneous-sparse structure and the equivalent MMV recovery formats of \eqref{MMX} are:
\begin{equation}\label{VC1}
    \mathbf{\overline{y}}_i = \mathbf{A}_i \mathbf{\overline{x}}^z_i + \mathbf{\overline{n}}_i, ~~~~ i=1,2,\hdots,L. 
\end{equation}
This common-support structure enables the recovery algorithm to jointly exploit the information contained in multiple observation vectors, thereby improving SR accuracy.}
The overall system model of (\ref{VC1}) encompassing the $L$ training frames, which can be represented as
\begin{equation}
    \mathbf{\overline{y}_g} = \mathbf{\ddot{A}}\mathbf{\overline{x}}^z_{\mathbf{g}} + \mathbf{\overline{n}_g},
\end{equation}
where we have $\mathbf{\overline{y}_g} = [\mathbf{\overline{y}}^T_1 ~\mathbf{\overline{y}}^T_2~ \hdots ~\mathbf{\overline{y}}^T_L]^T \in \mathbb{C}^{N_z L\times 1}, \mathbf{\overline{x}}^z_{\mathbf{g}} = [(\mathbf{\overline{x}}^z_1)^T~ (\mathbf{\overline{x}}^z_2)^T \hdots~(\mathbf{\overline{x}}^z_L)^T]^T \in \mathbb{C}^{N_z L\times 1},  \mathbf{\overline{n}_g} = [\mathbf{\overline{n}}^T_1~ \mathbf{\overline{n}}^T_2~ \hdots~ \mathbf{\overline{n}}^T_L]^T \in \mathbb{C}^{N_z L \times 1}$, and the equivalent measurement matrix is $\mathbf{\ddot{A}} = \mathrm{blkdiag}\left( \{ \mathbf{A}_i \}_{i=1}^{L}\right) \in \mathbb{C}^{N_z L \times N_z L}$. Fig. \ref{a}\subref{1n} shows the simultaneous-sparse and group-sparse ($\mathbf{\overline{x}}^z_{\mathbf{g}}$) structure of the over-sampled FD O-OFDM symbol vector. The subsequent section describes the traditional CS-based techniques harnessed for CE and SR in the MMV-based O-OFDM sparse VLC system. 

\section{CS-based CE-aided SR in MMV-based sparse VLC system model}

\textcolor{black}{This section introduces the approximately sparse multipath VLC channel model and the CE methodology adopted in this work. In indoor VLC environments, optical propagation is predominantly diffusive, and the emitted waveform typically undergoes multiple reflections before impinging on the receiver. Consequently, the received signal comprises a dominant specular component together with several weaker diffuse contributions. The resulting discrete-time multipath channel can be expressed as \cite{wilson2011scheduling,qian2017pilot}
\begin{equation}\label{resp2}
h(n)=\sum_{l=0}^{L_h-1} h(l)\delta\big(n-\nu(l)\big),
\end{equation}
where $h(l)$ denotes the VLC CIR coefficient and $\nu(l)$ is the propagation delay associated with the $l^\mathrm{th}$ path. In practice, a substantial portion of the received optical power is carried by a small number of dominant components, whereas the remaining paths contribute negligibly due to higher attenuation \cite{zhang2016enhancing,zhao2014compressed,minn2002investigation}. This structure motivates treating the multipath VLC channel as approximately sparse \cite{niaz2016compressed,zhang2016enhancing,shi2020adaptive,du2016channel,shen2019experimental,xiaoli2021research,9235008,zhao2014compressed,saxena2025multiple}. Accordingly, insignificant taps may be removed via thresholding, setting their gains to zero to improve estimation performance \cite{zhang2016enhancing}. Moreover, since longer-delay reflections experience stronger loss, the tap magnitudes typically decrease with delay \cite{wilson2011scheduling}. Therefore, the TD coefficient of the $l^\mathrm{th}$ path is modeled using an exponential power-delay profile \cite{zhao2013channel,zhao2014compressed,jungnickel2002physical,wu2012channel,saxena2025multiple}
\begin{equation}\label{rep3}
h(l)=\frac{e^{-l t_s/\tau}}{\sum_{l=0}^{L_h-1} e^{-l t_s/\tau}},
\end{equation}
where $\tau$ denotes the channel delay spread, governed by the room geometry and surface reflectivity, and $t_s$ is the sampling period of the O-OFDM TD samples. Typically, $\tau \in [0.5t_s,1.5t_s]$ \cite{jungnickel2002physical,5206382,wilson2011scheduling}. In addition, taps with magnitude below $10^{-8}$ are treated as zero \cite{zhao2014compressed}, resulting in an $s$-sparse VLC CIR with at most $s$ nonzero coefficients, where $s \ll L_h$.}
Note that the discrete TD framework is represented by (\ref{resp2}). Moreover, the normalized samples of the TD multipath VLC CIR are given by (\ref{rep3}). Consequently, (\ref{resp2}) and (\ref{rep3}) model the sparse multipath VLC CIR. 

Next, we describe the traditional CE techniques that exploit the simultaneous group-sparsity of the multipath CIR vector $\mathbf{h}$ and the SR methods that leverage the similar sparse structure of the O-OFDM signal $\mathbf{\overline{x}}^z$ within the VLC system framework. To achieve the goal of CE, let the total number of pilot subcarriers be $N_P$, and the total number of data subcarriers be $N_D$, where, for simplicity, we have assumed that $N_P = N_D = N$. The over-sampled FD pilot symbols are denoted as $\mathbf{\overline{x}}_P^z = [\tilde{x}_0, \tilde{x}_1, \hdots, \tilde{x}_{N_z^P-1}]^T\in \mathbb{C}^{N_z^P \times 1}$ and the corresponding compressed pilot symbol is expressed as $\mathbf{\overline{x}}_P^c = \mathbf{\Phi} \mathbf{\overline{x}}_P^z \in \mathbb{C}^{M_P \times 1}$, where $N^P_z = N_Pl_z$ and $M_P = \alpha N^P_z = M$. The CTF of VLC in FD is denoted by $\mathbf{\overline{h} = Qh}$, where $\mathbf{Q} \in \mathbb{C}^{M \times L_h}$, represents the pruned DFT matrix associated with $Q_{k,l} = e^{-j\left( \frac{2\pi kl}{M}\right)}  $, for $ 0\leq l \leq L_h-1,\text{and~} 0\leq k \leq M-1$. Thus, the FFT processing at the receiver in the case of reduced-PAPR VLC channel is
\begin{equation}
\begin{aligned}
    \mathbf{\overline{y}}_P^c &= \text{diag}\{\mathbf{\overline{h}} \}\mathbf{\overline{x}}_P^c + \mathbf{\overline{n}}_P^c = \text{diag}\{\mathbf{\overline{x}}_P^c \}\mathbf{\overline{h}} + \mathbf{\overline{n}}_P^c \\
    & = \text{diag}\{\mathbf{\Phi \overline{x}}_P^z \}\mathbf{Qh} + \mathbf{\overline{n}}_P^c.
\end{aligned}
\end{equation}
The resultant output pilot symbol $\mathbf{\overline{y}}_P \in \mathbb{C}^{N^P_z \times 1}$ after the despreading of the ZC sequence is given by
\begin{equation}\label{CE1}
    \mathbf{\overline{y}}_P = \mathbf{\Phi}^H\mathbf{\overline{y}}_P^c = \mathbf{\Phi}^H \text{diag}\{\mathbf{\Phi \overline{x}}_P^z \}\mathbf{Qh} + \mathbf{\Phi}^H\mathbf{\overline{n}}^c_P = \mathbf{Bh} + \mathbf{\overline{n}}_P.
\end{equation}
\textcolor{black}{Since the indoor VLC channel is assumed to remain quasi-static over the $L$ pilot observations, the transmitter-receiver geometry and the dominant reflection paths are expected to be unchanged within this interval. Accordingly, the indices of the significant CIR coefficients are also assumed to remain invariant. As a result, the VLC CIR vectors exhibit a common sparse support, leading to a simultaneous-sparse MMV formulation, as illustrated in Fig. \ref{a}\subref{2}. The resulting model is written as
\begin{equation}\label{CE2}
    \mathbf{\overline{y}}_{P,i}= \mathbf{B}_i\mathbf{h}_i + \mathbf{\overline{n}}_{P,i}, ~~~~ i=1,2,\hdots,L.
\end{equation}}
The cumulative system model of (\ref{CE2}) over $L$ training frames can be represented as
\begin{equation}\label{CE3}
    \mathbf{\overline{y}}^{P}_{\mathbf{g}} = \mathbf{\ddot{B}}\mathbf{h_g} + \mathbf{\overline{n}}^P_{\mathbf{g}}, 
\end{equation}
where $\mathbf{\overline{y}}^{P}_{\mathbf{g}} = [\mathbf{\overline{y}}^T_{P,1}~ \mathbf{\overline{y}}^T_{P,2}~ \hdots~ \mathbf{\overline{y}}^T_{P,L}]^T \in \mathbb{C}^{LN^P_z \times 1}$, $\mathbf{h_g} = [\mathbf{h}^T_1 ~\mathbf{h}^T_2~ \hdots~ \mathbf{h}^T_L]^T \in \mathbb{R}^{L_h L \times 1}_{+}$, $ \mathbf{\overline{n}}^P_{\mathbf{g}} = [\mathbf{\overline{n}}^T_{P,1}~ \mathbf{\overline{n}}^T_{P,2}~ \hdots~ \mathbf{\overline{n}}^T_{P,L}] \in \mathbb{C}^{LN^P_z\times 1}$, and the equivalent measurement matrix is $\mathbf{\ddot{B}} = \mathrm{blkdiag}\left( \{ \mathbf{B}_i \}_{i=1}^{L}\right)
\in \mathbb{C}^{L N^P_z \times L_h L}$. Consequently, the objective is to perform CE by leveraging the group-sparsity of the concatenated multipath CIR vector $\mathbf{h_g}$ with the aid of the given output pilot vector $\mathbf{\overline{y}}^{P}_{\mathbf{g}}$. Let the estimated sparse VLC CIR matrix using the CS-based technique be denoted by $\mathbf{\widehat{h}_g}$, where we have $\mathbf{\widehat{h}_g} = [\mathbf{\widehat{h}}^T_1~ \mathbf{\widehat{h}}^T_2~ \hdots~ \mathbf{\widehat{h}}^T_L]^T \in \mathbb{R}^{L_h L \times 1}_{+}$. Again, the estimated FD VLC CTF is given by $\mathbf{\overline{h}}_i = \mathbf{Q\widehat{h}}_i$, where $i = 1, 2, \hdots, L$. The over-sampled FD data symbols are denoted as $\mathbf{\overline{x}}_D^z = [\tilde{x}_0, \tilde{x}_1, \hdots, \tilde{x}_{N_z^D-1}]^T\in \mathbb{C}^{N_z^D \times 1}$ and the corresponding compressed data symbol is expressed as $\mathbf{\overline{x}}_D^c = \mathbf{\Phi} \mathbf{\overline{x}}_D^z \in \mathbb{C}^{M_D \times 1}$, where $N^D_z = N_Dl_z$ and $M_D = \alpha N^D_z = M$.
The resultant FD output data vector $\mathbf{\overline{y}}^c_{D,i} \in \mathbb{C}^{M_D \times 1} $ follows from (\ref{MMX}) as
\begin{equation}\label{MMVLC1}
\begin{aligned}
    \mathbf{\overline{y}}^c_{D,i} = \text{diag}\{ \mathbf{\overline{h}}_i\} \mathbf{\overline{x}}_{D,i}^c + \mathbf{\overline{n}}^c_{D,i} = \text{diag}\{ \mathbf{\overline{h}}_i\} \mathbf{\Phi}  \mathbf{\overline{x}}_{D,i}^z + \mathbf{\overline{n}}^c_{D,i}
\end{aligned}
\end{equation}
The resultant output data symbol $\mathbf{\overline{y}}_D \in \mathbb{C}^{N^D_z \times 1}$ after the despreading of the ZC sequence is given by
\begin{equation}\label{MMVLC21}
\begin{aligned}
    \mathbf{\overline{y}}_{D,i} & = \mathbf{\Phi}^H\mathbf{\overline{y}}^c_{D,i} = \mathbf{\Phi}^H\text{diag}\{ \mathbf{\overline{h}}_i\} \mathbf{\Phi}  \mathbf{\overline{x}}_{D,i}^z + \mathbf{\Phi}^H\mathbf{\overline{n}}^c_{D,i} \\
&=\mathbf{\widehat{A}}_{D,i}\mathbf{\overline{x}}_{D,i}^z + \mathbf{\overline{n}}_{D,i},
\end{aligned}
\end{equation}
where the sparse FD data vector is $\mathbf{\overline{x}}_{D,i}^z \in \mathbb{C}^{N^D_z \times 1}$, and the estimated measurement matrix is $\mathbf{\widehat{A}}_{D,i} = \mathbf{\Phi}^H\text{diag}\{ \mathbf{\overline{h}}_i\} \mathbf{\Phi} \in \mathbb{C}^{N^D_z \times N^D_z}$ for $i=1,2,\hdots,L$. Since the sparsity profile of $\mathbf{\overline{x}}_{D,i}^z$ is the same for $L$ different measurements vectors, the MMV representation of (\ref{MMVLC1}) is  
\begin{equation}\label{MMVLC2}
    \mathbf{\overline{y}}^D_{\mathbf{g}} = \mathbf{\ddot{A}}_D\mathbf{\overline{x}}^z_{\mathbf{g}} + \mathbf{\overline{n}}^D_{\mathbf{g}},
\end{equation}
where $\mathbf{\overline{y}}^D_{\mathbf{g}} = [\mathbf{\overline{y}}^T_{D,1}~ \mathbf{\overline{y}}^T_{D,2}~\hdots~ \mathbf{\overline{y}}^T_{D,L}]^T \in \mathbb{C}^{LN^D_z \times 1}, \mathbf{\overline{x}}^z_{\mathbf{g}} = [(\mathbf{\overline{x}}^z_{D,1})^T~ (\mathbf{\overline{x}}^z_{D,2})^T~ \hdots~(\mathbf{\overline{x}}^z_{D,L})^T]^T \in \mathbb{C}^{LN^D_z \times 1}, \mathbf{\overline{n}}^D_{\mathbf{g}} = [\mathbf{\overline{n}}^T_{D,1}~ \mathbf{\overline{n}}^T_{D,2}~ \hdots~ \mathbf{\overline{n}}^T_{D,L}]^T \in \mathbb{C}^{LN^D_z\times 1}$, and the equivalent estimated measurement matrix is $\mathbf{\ddot{A}}_D = \mathrm{blkdiag}\left(\{\mathbf{\widehat A}_{D,i} \}_{i=1}^{L}\right) \in \mathbb{C}^{LN^D_z \times LN^D_z}$. 
The objective of this procedure is to accomplish SR, exploiting the group-sparsity of the FD data vector $\mathbf{\overline{x}}^z_{\mathbf{g}}$, from the FD output data vector $\mathbf{\overline{y}}^D_{\mathbf{g}}$. As the O-OFDM-VLC system model outlined in (\ref{MMVLC1}) is effectively transformed into $M$-parallel flat-fading systems, one-tap equalization is sufficient for the retrieval of the FD data symbols $\mathbf{\overline{x}}_{D,i}^c$ and then subsequently $\mathbf{\overline{x}}_{D,i}^z$. The signal $\mathbf{\overline{x}}_{D,i}^z$ is then directed to a down-sampler and demapper/demodulator, which converts the complex estimated symbol $\mathbf{\overline{x}}_{D,i}^z$ into the appropriate bits, allowing the transmitted data to be recovered. The following section outlines the traditional LS and LMMSE-based CE and SR methods utilized to estimate the multipath VLC CIR $\mathbf{h_g}$ and sparse FD data $\mathbf{\overline{x}}^z_{\mathbf{g}}$. 
\subsection{Conventional VLC CE-aided SR techniques}
The CE process using the LS approach is stated as follows
\begin{equation}\label{Papr6}
    \mathbf{\widehat{h}}_{i,\text{LS}} = \arg \min _{\mathbf{h}_i} || \mathbf{\overline{y}}_{P,i}-\mathbf{B}_i\mathbf{h}_i||_2^2, ~~ i=1,2,\hdots,L. 
\end{equation}
Calculating the derivative of the objective function given in (\ref{Papr6}) relative to the $\mathbf{h}_i$ and equating it to zero results in the renowned LS solution, given as $\mathbf{\widehat{h}}_{i,\text{LS}} = (\mathbf{B}^H_i \mathbf{B}_i)^{-1} \mathbf{B}^H_i \mathbf{\overline{y}}_{P,i}$ \cite{kay1993fundamentals}. The CTF is formulated as $\mathbf{\overline{h}}_{i,\text{LS}} = \mathbf{Q}\mathbf{\widehat{h}}_{i,\text{LS}}$. Similarly, one can perform LS-based SR from (\ref{MMVLC21}) using the estimated CTF to obtain the FD data vector, which is given as $\mathbf{\overline{x}}^z_{D,i,\text{LS}} = (\mathbf{\widehat{A}}^H_{D,i} \mathbf{\widehat{A}}_{D,i})^{-1} \mathbf{\widehat{A}}^H_{D,i} \mathbf{\overline{y}}_{D,i}$ \cite{kay1993fundamentals}. Alternatively, the mean square error (MSE) of CE is minimized via the LMMSE estimator, which leverages a linear estimator defined as follows $\mathbf{\widehat{h}}_{i,\text{LMMSE}} = \mathbf{L\overline{y}}_{P,i}$.  The associated MSE is defined as $\mathbb{E} \left[ || \mathbf{\widehat{h}}_{i,\text{LMMSE}}-\mathbf{h}_i ||^2_2 \right]$. Minimizing the MSE in the context of $\mathbf{L}$ yields $\mathbf{L} = \mathbf{R}_{\mathbf{h}_i\mathbf{\overline{y}}_{P,i}} \mathbf{R}_{\mathbf{\overline{y}}_{P,i}\mathbf{\overline{y}}_{P,i} }^{-1}$. Thus, the multipath CIR $\mathbf{h}_i$ obtained using the conventional LMMSE CE is given as follows $ \mathbf{\widehat{h}}_{i,\text{LMMSE}} = \mathbf{R}_{\mathbf{h}_i\mathbf{\overline{y}}_{P,i}} \mathbf{R}_{\mathbf{\overline{y}}_{P,i}\mathbf{\overline{y}}_{P,i}}^{-1}\mathbf{\overline{y}}_{P,i},$
where $\mathbf{R}_{\mathbf{h}_i\mathbf{\overline{y}}_{P,i}} =  \mathbb{E}[\mathbf{h}_i\mathbf{\overline{y}}_{P,i}^H]$ denotes the cross-covariance matrix of the output pilot symbol $\mathbf{\overline{y}}_{P,i}$ with multipath CIR $\mathbf{h}_i$, and the auto-covariance of output pilot vector $\mathbf{\overline{y}}_{P,i}$ determined as $\mathbf{R}_{\mathbf{\overline{y}}_{P,i}\mathbf{\overline{y}}_{P,i}} = \mathbb{E}[\mathbf{\overline{y}}_{P,i}\mathbf{\overline{y}}_{P,i}^H]$. Upon incorporating the quantities $\mathbf{R}_{\mathbf{h}_i\mathbf{\overline{y}}_{P,i}}$ and $\mathbf{R}_{\mathbf{\overline{y}}_{P,i}\mathbf{\overline{y}}_{P,i}}$, the LMMSE estimator can be explicitly described as
\begin{equation}
  \mathbf{\widehat{h}}_{i,\text{LMMSE}} = \left(\mathbf{B}^H_i \mathbf{R}_{P,i}^{-1} \mathbf{B}_i + \mathbf{R}_{\mathbf{h}_i\mathbf{h}_i}^{-1} \right)^{-1} \mathbf{B}^H_i \mathbf{R}_{P,i}^{-1} \mathbf{\overline{y}}_{P,i},  
\end{equation}
where $i = 1, 2, \hdots, L$. In the above, $\mathbf{R}_{\mathbf{h}_i\mathbf{h}_i} = \mathbb{E}[\mathbf{h}_i\mathbf{h}^H_i] $ represents the \textit{a priori} covariance matrix of the multipath VLC CIR $\mathbf{h}_i$ and the covariance matrix of the noise is given by $\mathbf{R}_{P,i} = \mathbb{E}[\mathbf{\overline{n}}_{P,i} \mathbf{\overline{n}}_{P,i}^H]$. Furthermore, the estimated CTF is obtained as $\mathbf{\overline{h}}_{i,\text{LMMSE}} = \mathbf{Q}\mathbf{\widehat{h}}_{i,\text{LMMSE}}$. Along similar lines, using (\ref{MMVLC21}) and the estimated CTF, one can perform LMMSE-based SR to obtain the estimated FD data vector given as
\begin{equation}
  \mathbf{\overline{x}}_{D,i,\text{LMMSE}}^z = 
  \left( \mathbf{\widehat{A}}^H_{D,i} \mathbf{R}_{D,i}^{-1} \mathbf{\widehat{A}}_{D,i} + \mathbf{R}_{\mathbf{\overline{x}}_{D,i}^z{\mathbf{\overline{x}}_{D,i}^z} }^{-1}\right)^{-1}\mathbf{\widehat{A}}^H_{D,i} \mathbf{R}_{D,i}^{-1} \mathbf{\overline{y}}_{D,i}, 
\end{equation}
where $i = 1, 2, \hdots, L$, $\mathbf{R}_{\mathbf{\overline{x}}_{D,i}^z{\mathbf{\overline{x}}_{D,i}^z}} = \mathbb{E}[(\mathbf{\overline{x}}_{D,i}^z)(\mathbf{\overline{x}}_{D,i}^z)^H]$ is the \textit{a priori} covariance matrix of the FD data vector $\mathbf{\overline{x}}_{D,i}^z$ and $\mathbf{R}_{D,i} = \mathbb{E}[\mathbf{\overline{n}}_{D,i}\mathbf{\overline{n}}_{D,i}^H]$ is the noise covariance matrix \cite{kay1993fundamentals}.

It is important to note that traditional methods such as LMMSE as well as LS do not consider the sparse characteristics of the multipath CIR $\mathbf{h}_i$, the FD data vector $\mathbf{\overline{x}}^z_{D,i}$, and group-sparsity of the CIR $\mathbf{h_g}$ and the FD data vector $\mathbf{\overline{x}}_g^z$. Moreover, conventional estimation requires $N_z^P \geq L_h$ for the existence of a unique solution, meaning that the number of pilots is required to be higher than or equal to the CIR length. Consequently, sparse estimation techniques are more suitable for the CE problems derived in (\ref{CE3}) and the SR problems derived in (\ref{MMVLC2}), since they take into consideration the simultaneous group-sparsity of the CIR $\mathbf{h_g}$ and of the O-OFDM signal $\mathbf{\overline{x}}^z_{\mathbf{g}}$, which result in a notable pilot overhead reduction and SR accuracy improvement. In this regard, the modified $l_0$-norm minimization-based cost function for simultaneous sparse CE is formulated as
\begin{equation}\label{bn1}
\begin{aligned}
& \underset{\mathbf{h_g}}{\mathrm{minimize}} ~~~ {\Vert \mathbf {\mathbf{h_g}}\Vert }_0
\\
& \text{subject to} ~~~ {\Vert { \mathbf{\overline{y}}^{P}_{\mathbf{g}}} - \mathbf{\ddot{B}\mathbf{h}_g}\Vert }_0^2\leq \text{\textxi}_1.
\end{aligned}
\end{equation}
Similarly, the modified $l_0$-norm minimization-based cost function for the simultaneous sparse SR is expressed as
\begin{equation}\label{GOMP1}
\begin{aligned}
& \underset{\mathbf{\overline{x}}^z_{\mathbf{g}}}{\mathrm{minimize}} ~~~ {\Vert \mathbf{\overline{x}}^z_{\mathbf{g}}\Vert }_0
\\
& \text{subject to} ~~~ {\Vert {\mathbf{\overline{y}}^D_{\mathbf{g}}} - \mathbf{\ddot{A}}_D\mathbf{\overline{x}}^z_{\mathbf{g}}\Vert }_0^2\leq \text{\textxi}_2, 
\end{aligned}
\end{equation}
where $\text{\textxi}_1$ and $\text{\textxi}_2$ are tunable parameters that rely on the noise power, which are defined as $\text{\textxi}_1 = \mathrm{Tr}\left( \mathbf{R}_P \right)$, $\text{\textxi}_2 = \mathrm{Tr}\left( \mathbf{R}_D \right)$ with $\mathbf{R}_P = \mathbb{E}\left[(\mathbf{\overline{n}}^P_{\mathbf{g}})(\mathbf{\overline{n}}^P_{\mathbf{g}})^H\right]$ and $\mathbf{R}_D = \mathbb{E}\left[(\mathbf{\overline{n}}^D_{\mathbf{g}})(\mathbf{\overline{n}}^D_{\mathbf{g}})^H\right]$ denoting the noise covariance matrices. The significant benefit of this sparse CS-based technique is that it requires a remarkably reduced number of observations for recovering the respective sparse vectors, thanks to the powerful CS method \cite{4472240}. The subsequent section defines a GOMP-based technique conceived for sparse SR-aided CE relying on a significantly reduced number of measurement vectors, i.e., $N_z^P << L_h$.
\begin{algorithm}[t]
\DontPrintSemicolon
  \KwIn{Pilot vector received $\mathbf{\overline{y}}^P_{\mathbf{g}} \in \mathbb{C}^{N^P_zL \times 1}$, stopping threshold \textctc, measurement matrix $\mathbf{\ddot B} \in \mathbb{C}^{N^P_zL \times L_{h}L}$}
  \KwOut{Estimated group CIR sparse $\mathbf{h_g}$}
  \textbf{Initialization:} Residue $\mathbf{r}_{-1}=\mathbf{0}_{LN^P_z \times 1}$, index set $\text{\textUpsilon}$ = [], iteration number $t=1$, $\mathbf{r}_{0}=\mathbf{\overline{y}}^P_{\mathbf{g}}$,  $\mathbf{\ddot{B}}^{\text{\textUpsilon}} = $[], $\mathbf{\widehat{h}}_{\text{GOMP}} = \mathbf{0}_{LL_h \times 1}$ 
  
  \While{$( \vert ~||\mathbf{r}_{t-1}||^2_2 - ||\mathbf{r}_{t-2}||^2_2~ \vert \geq \text{\textctc} )$} 
    {   
       $ \Psi = \mathbf{\ddot{B}}^H \mathbf{r}_{t-1}$\\
       
       $ \varphi =\mathrm {diag}\left ({{\Psi\Psi ^{H}}}\right)$ 
       
       $j = \mathop{\mathrm {arg\,\,max}}\limits_{k = 0,1,\ldots, L_h-1} \sum _{l=1}^{L} \varphi \left [{ (l-1)L_h + k }\right]$
       
       $\mathcal {J} = \big \{ \left [{ (l-1) L_h }\right] + j \big \}_{l=1}^{L}$
       
       $\text{\textUpsilon} = \text{\textUpsilon} \cup \mathcal {J}$
       
       $\mathbf{\ddot{B}}^{\text{\textUpsilon}} = \mathbf{\ddot B}(:,\text{\textUpsilon})$
       
       $\mathbf{h}^t_{\mathbf{g}} = (\mathbf{\ddot B}^{\text{\textUpsilon}})^\dagger \mathbf{\overline{y}}^P_{\mathbf{g}}$
       
       $\mathbf{r}_t = \mathbf{\overline{y}}^P_{\mathbf{g}} - \mathbf{\ddot B}^{\text{\textUpsilon}}\mathbf{h}^t_{\mathbf{g}}$
       
       $t = t + 1$
    }\textbf{end} \\
    \textbf{return:}$~~\mathbf{\widehat{h}}_{\mathbf{g,}\text{GOMP}}(\text{\textUpsilon}) = \mathbf{h}^t_{\mathbf{g}}$
\caption{GOMP-assisted group sparse CE in the O-OFDM-VLC system}
\label{Algo1}
\end{algorithm}
\subsection{GOMP-based sparse CE-aided sparse SR technique}
The GOMP algorithm is an iterative and greedy algorithm that has been designed for group-sparse recovery from the OMP proposed in \cite{4385788}. The steps involved in GOMP-based sparse CE are shown in Algorithm \ref{Algo1}. The following procedures represent the key differences with respect to the standard OMP method of \cite{4385788}. Steps $3$, $4$, and $5$ of each iteration of GOMP are responsible for extracting the group index, denoted by $j$, having the highest correlation with the preceding residue $\mathbf{r}_{t-1}$. Next, in Step $6$, all the $L$ indices that correspond to the selected group are chosen. This distinguishing feature enables GOMP to choose $L$ columns from the matrix $\mathbf{\ddot{B}}$ during each iteration, which corresponds to a group, culminating in a group-sparse estimation. As a result, GOMP performs better than OMP, which chooses a single column from the dictionary matrix for each iteration. Algorithm \ref{Algo1} illustrates the method of group-sparse CE utilizing GOMP. The estimated CTF is expressed as $\mathbf{\overline{h}}_{\mathbf{g},\text{GOMP}} = \mathbf{\ddot Q}\mathbf{\widehat{h}}_{\mathbf{g},\text{GOMP}}$, where we have $\mathbf{\ddot{Q}} = \mathrm{blkdiag}\left(\{\mathbf{Q} \}_{i=1}^{L}\right) \in \mathbb{C}^{LM \times LL_h}$. Thus, the estimated CTF of each frame is $\mathbf{\overline{h}}_{i,\text{GOMP}} = \bigg \{ \mathbf{\overline{h}}_{\mathbf{g,\text{GOMP}}} \left[ (i-1)L + n \right] \bigg \}^{M}_{n=1}$, the estimated measurement matrix is $\mathbf{\widehat{A}}_{D,i} = \mathbf{E}_D\text{diag}\{\mathbf{\overline{h}}_{i,\text{GOMP}}\}\mathbf{E}_D^H $ and thus, the equivalent estimated measurement matrix is $\mathbf{\ddot{A}}_D = \mathrm{blkdiag}\left(\{\mathbf{\widehat A}_{D,i} \}_{i=1}^{L}\right)$ for $i = 1, 2, \hdots, L$. In an analogous manner, one can perform GOMP-based sparse SR from (\ref{MMVLC2}) and (\ref{GOMP1}) using the estimated CFR $\mathbf{\overline{h}}_{i,\text{GOMP}}$, which yields the estimated sparse FD data vector $\mathbf{\overline{x}}^z_{\mathbf{g}}$.

Nonetheless, it is essential to recognize that the proposed GOMP scheme suffers from the same limitations as OMP. These limitations involve sensitivity to the dictionary matrix, to error propagation, and to the stopping threshold selection \cite{srivastava2021bayesian,srivastava2021sparse}. To surmount these shortcomings of GOMP, we present a GBL-based sparse SR technique in the following section that exhibits robust convergence and does not require any regularization/tuning.
\section{Bayesian Learning (GBL)-based group sparse CE and SR in O-OFDM VLC systems}
The proposed GBL framework has evolved from the classic Bayesian principle. In the VLC system, this entails initiating the procedure by allocating a parameterized Gaussian distributed prior to the sparse VLC channel vector $\mathbf{h}_i \in \mathbb{R}^{L_{h} \times 1}$, which is unknown in general, and it is modeled as \cite{srivastava2021sparse}:
\begin{equation} 
f(\mathbf {h}_i; \mathbf \Gamma)=\prod _{k=0}^{L_{h}-1}\frac{1}{(2\pi \gamma _k)^{-\frac{1}{2}}} \exp \Bigg (-\displaystyle \frac{|h_i(k)|^2}{2\gamma _k}\Bigg).  
\end{equation}
Within this framework, the hyperparameter is given by ${\gamma _k}$, $0\leq k \leq L_h-1$, and the hyperparameter matrix is $\boldsymbol{\Gamma } = \mathrm {diag}\big \{ \gamma _{k} \big \}_{k=0}^{L_h-1} \in \mathbb {R}^{L_h \times L_h}_{+}$. It is important to highlight that the \textit{a priori} covariance $\mathbf{R}_h$ of the sparse VLC channel vector $\mathbf{h}_i$ follows $\mathbf{R}_h = \mathbf \Gamma$. Furthermore, it is assumed that $\mathbf{R}_h$ is not known at the initial stage. Given that the sparsity pattern of $\mathbf{h}_i$ is identical for $1 \leq i \leq L$, the corresponding prior for the group sparse VLC CIR $\mathbf{h_g}$ is given by
\begin{equation} \label{GSBL1}
f(\mathbf {h_g}; \mathbf \Gamma)=\prod _{l=1}^{L}\prod _{k=0}^{L_{h}-1}\frac{1}{(2\pi \gamma _k)^{-\frac{1}{2}}} \exp \Bigg (-\displaystyle \frac{|h_l(k)|^2}{2\gamma _k}\Bigg).  
\end{equation}
In order to capitalize on the group sparsity of the VLC channel $\mathbf{h_g}$, the hyperparameter $\gamma _k$ is set to every element of the $k$th group, so that the index varies from $\big \{ [(l-1)L_h]+k \big \}_{l=1}^{L}
$. This results in the $k$th group being entirely comprised of either zeros or non-zeros. Utilizing this unique group-sparse framework, the proposed GBL method estimates the $[LL_h \times 1]$-element vector $\mathbf{h_g}$ via as few as $L_h$ hyperparameters. Furthermore, the \textit{a priori} covariance matrix $\mathbf{\ddot{R}}_h$ of the group-sparse vector $\mathbf{h_g}$ is given by $\mathbf{\ddot{R}}_{h} = \left ({\mathbf {I}_{L} \otimes \mathbf{\Gamma }}\right)$. The mean estimate ${\boldsymbol{\mu }}_{\mathbf g} \in \mathbb {C}^{LL_{h} \times 1}$ and the corresponding error covariance matrix ${\boldsymbol{\Sigma }}_{\mathbf g} \in \mathbb {C}^{LL_{h} \times LL_{h}}$ of the VLC CIR $\mathbf{h_g}$ is given by 
\begin{equation}\label{App3}
    {\boldsymbol{\mu }}_{\mathbf g}= {\boldsymbol{\Sigma }}_{\mathbf g} \mathbf {\ddot{B}}^H \mathbf {R}_P^{-1} \mathbf{\overline{y}}^P_\mathbf{g}, ~~~ {\boldsymbol{\Sigma }}_{\mathbf g} =\Big({\mathbf {\ddot{B}}^H} \mathbf{R}_P^{-1} {\mathbf {\ddot{B}}}+\left(\mathbf {I}_{L} \otimes \mathbf {\widehat{\Gamma} }^{-1}\right)\Big)^{-1},
\end{equation}
where $\mathbf{R}_P \in \mathbb{C}^{LN^P_z \times LN^P_z}$. Consequently, in order to derive the mean estimate ${\boldsymbol{\mu }}_{\mathbf g}$, the hyperparameter matrix $\mathbf \Gamma$ has to be estimated. Furthermore, we observe from (\ref{GSBL1}) that as the hyperparameter obeys $\gamma _{k} \rightarrow 0$, each of the associated elements of the $k$th group also tends to zero \cite{srivastava2021sparse}. To accomplish this objective, it is preferable to select the log-Bayesian evidence $\log \left [{ f(\mathbf{\overline{y}}^P_{\mathbf{g}}; \mathbf \Gamma)}\right]$  maximization matrix $\mathbf{\widehat\Gamma}$, which results in $\log \left [{ f(\mathbf{\overline{y}}^P_{\mathbf{g}}; \mathbf \Gamma)}\right] = c_1 - \log \left[ \det\left( \mathbf{\Sigma}_P\right)\right]-\left(\mathbf{\overline{y}}^P_{\mathbf{g}}\right)^H\left(\mathbf{\Sigma}_P \right)^{-1}\left(\mathbf{\overline{y}}^P_{\mathbf{g}}\right) $, where the $c_1 = -LN^P_z \log(\pi)$ is a constant and the received pilot covariance matrix is $\mathbf{\Sigma}_P = \mathbf{R}_P + \mathbf{\ddot{B}}\left(\mathbf{I}_L \otimes \mathbf\Gamma \right)\mathbf{\ddot{B}}^H \in \mathbb{C}^{LN^P_z \times LN^P_z}$. The resultant optimization problem is intractable, since maximizing the evidence of the log-Bayesian with regard to the hyperparameter matrix $\mathbf{\Gamma}$ results in a non-concave problem \cite{srivastava2021sparse}. In this case, maximizing the cost function at each iteration may be accomplished using the expectation-maximization (EM) method. Furthermore, finding a local optimum is also guaranteed by the EM method. Consequently, the proposed GBL technique utilizes the EM method for group-sparse CSI estimation in the O-OFMD-VLC system. We next infer the ensuing stages of this process. The complete information set is given as $\left \{{\mathbf{\overline{y}}^P_{\mathbf{g}}
, \mathbf{h_g}}\right \}$, so that the hyperparameter estimate in the $(j-1)$st EM-iteration is $\widehat {\boldsymbol{\Gamma }}^{(j-1)} = \mathrm {diag}\left \{{ \widehat {\gamma }_{k}^{(j-1)} }\right \}_{k=0}^{L_h-1} \in \mathbb {R}^{L_h \times L_h}_{+}$. The process of revising the estimate $\widehat {\boldsymbol{\Gamma }}^{(j)}$ in the $j$th EM-iteration is outlined in Theorem-$1$.
\begin{theorem}
\textit{In the $j$th EM iteration, the hyperparameter update $\widehat {\gamma }_{k}^{(j)}$ that maximizes the conditional-expectation of the log-likelihood function for the given $\widehat {\gamma }_{k}^{(j-1)}$, where $0\leq k\leq L_h-1$, corresponding to the complete information set $\left \{{\mathbf{\overline{y}}^P_{\mathbf{g}}
, \mathbf{h_g}}\right \}$ is represented by
\begin{equation} 
\mathbf {\mathcal {L}}\left ({\boldsymbol{\Gamma }| \widehat {\boldsymbol{\Gamma }}^{(j-1)}}\right)= \mathbb {E}_{ \mathbf{h_g} \vert \mathbf{\overline{y}}^P_{\mathbf{g}}; \widehat {\boldsymbol{\Gamma }}^{(j-1)}} \bigg \{\log \left [{ f \left ({\mathbf{\overline{y}}^P_{\mathbf{g}}
, \mathbf{h_g}}; \boldsymbol{\Gamma } \right) }\right] \bigg \},
\end{equation}
is as follows:  
\begin{equation}\label{App6}
\begin{aligned}
\widehat{\gamma }_{k}^{(j)}& = \frac {1}{L} \sum _{l=1}^{L}  {\boldsymbol{\Sigma }}^{(j)}_{\mathbf{g}} \big [k+(l-1)L,k+(l-1)L \big]\\
&+ \frac{1}{L} \sum_{l=1}^{L} \left \vert{ \boldsymbol{\mu}^{(j)}_{\mathbf{g}}\big [k+(l-1)L\big] }\right \vert ^{2}, 
\end{aligned}
\end{equation}
where ${\boldsymbol{\Sigma }}^{(j)}_{\mathbf g} =\left[{\mathbf {\ddot{B}}^H} \mathbf{R}_P^{-1} {\mathbf {\ddot{B}}}+\left(\mathbf {I}_{L} \otimes \left(\mathbf {\widehat \Gamma }^{(j-1)}\right)^{-1}\right)\right]^{-1}$ and $\boldsymbol {\mu}^{(j)}_{\mathbf{g}} = {\mathbf{\Sigma_g}}^{(j)} \mathbf{\ddot B}^{H} {\mathbf {R}}_{P}^{-1} {\mathbf {\overline{y}}}^P_{\mathbf{g}}$.}
\end{theorem}

\textit{Proof}: The Appendix provides the proof.

The aforementioned EM method is repeated until $\parallel\widehat {\boldsymbol {\Gamma }}^{(j)} - \widehat {\boldsymbol {\Gamma }}^{(j-1)} \parallel ^{2}_{F}\leq \epsilon$ is satisfied or for a maximum of $m_{\max}$ iterations, whichever is met first. Thus, following convergence, the GBL-assisted group sparse VLC channel is derived from the \textit{a posteriori} mean, i.e., $\mathbf{\widehat h}_{\mathbf{g},\text{GBL}} =\boldsymbol{\mu}^{(j)}_{\mathbf{g}}$ and the corresponding error covariance matrix $\mathbf{\Sigma}^{(j)}_{\mathbf{g}}$ represents the resultant estimation uncertainty. The estimated CTF is denoted as $\mathbf{\overline{h}}_{\mathbf{g},\text{GBL}} = \mathbf{\ddot Q}\mathbf{\widehat h}_{\mathbf{g},\text{GBL}}$. Algorithm \ref{sfblms_algo} summarizes the group-sparse CE. Similarly, one can perform GBL-based sparse SR using the estimated CFR $\mathbf{\overline{h}}_{i,\text{GBL}}$, 
which results in the estimated sparse FD data vector $\mathbf{\overline{x}}^z_{\mathbf{g}}$. In this case, the estimated CFR of each frame is $\mathbf{\overline{h}}_{i,\text{GBL}} = \bigg \{ \mathbf{\overline{h}}_{\mathbf{g,\text{GBL}}} \left[ (i-1)L + n \right] \bigg \}^{M}_{n=1}$. The estimated measurement matrix is $\mathbf{\widehat{A}}_{D,i} = \mathbf{\Phi}^H\text{diag}\{\mathbf{\overline{h}}_{i,\text{GBL}}\}\mathbf{\Phi} $ and the equivalent estimated measurement matrix is $\mathbf{\ddot{A}}_D = \mathrm{blkdiag}\left(\{\mathbf{\widehat A}_{D,i} \}_{i=1}^{L}\right)$ for $i = 1, 2, \hdots, L$.

\begin{algorithm}[t]
\DontPrintSemicolon 
\KwIn{Measurement matrix $\mathbf{\ddot B} \in \mathbb{C}^{LN^P_z \times L_{h}L}$, received pilot vector $\mathbf{\overline{y}}^P_{\mathbf{g}} \in \mathbb{C}^{LN^P_z \times 1}$, stopping threshold $\epsilon$, noise covariance matrix $\mathbf {R}_P$, and $\ m_{\max}$}
\KwOut{Estimated group sparse CIR $\mathbf{h_g}$}
\textbf{Initialization:} $\widehat{{\gamma }}_k^{(0)} =1, \ \forall 0\leq k\leq L_{h}-1 \Longrightarrow \widehat{\mathbf {\Gamma }}^{(0)} = \mathbf {I}_{L_{h}}$, set
$\widehat{\mathbf {\Gamma }}^{(-1)} = \mathbf{0}$, and initialize counter $j = -1$  \nonumber

\While{$\left( j < m_{\max} ~~ \&\&~~ \left \Vert{\widehat{\boldsymbol{\Gamma }}^{(j+1)} - \widehat{\boldsymbol{\Gamma }}^{(j)}}\right \Vert^2_F > \epsilon \right)$}
{
 $j \leftarrow j+1$
 
 \textbf{E Step:} Evaluate \textit{a posteriori} covariance and mean given by (\ref{App3})
 
 
 
 \textbf{M Step:} Update the hyperparameter's estimates using (\ref{App6})
 
}\textbf{end}

\textbf{return:~~}{$\widehat{\mathbf{h}}_{\mathbf{g},\text{GBL}} = {\boldsymbol{\mu}}_{\mathbf{g}}^{(j)}$}
\caption{GBL-assisted group sparse CE in the O-OFDM-VLC system}
\label{sfblms_algo}
\end{algorithm}

According to \cite{1315936}, the GBL-based log-likelihood function exhibits fewer local maxima compared to schemes like FOCUSS \cite{gorodnitsky1997sparse}, which enables the GBL to converge to the global minima, leading to its improved convergence properties. Furthermore, the EM algorithm in GBL guarantees global convergence to a fixed log-likelihood point, with the optimization objective increasing in each iteration. Regarding convergence speed, recent research \cite{daskalakis2017ten} demonstrated that when approximating the parameters of a two-component Gaussian mixture model, the EM algorithm typically converges rapidly, within only $10$ iterations. Therefore, combining the robust convergence of the EM algorithm with the properties of the GBL-based log-likelihood yields improved performance for sparse channel estimation. Owing to space limitations, the derivation of the BCRLB related to the NMSE of the estimated group-sparse CSI matrix $\mathbf{\widehat{h}_g}$ is provided in the associated technical report \cite{shubPAPR}.
\subsection{Low Complexity GBL (LCGBL) Based Group-Sparse CE}
Based on the comprehensive discussions in the preceding section, the GBL algorithm necessitates the inversion of a \([L_hL \times L_hL]\)-dimensional matrix \({\boldsymbol{\Sigma }}_{\mathbf g}\) to compute the \textit{a posteriori} covariance, leading to significant computational complexity. To address this issue and streamline the process, this subsection introduces the LCGBL algorithm tailored for group-sparse CE. The covariance matrix \({\boldsymbol{\Sigma }}^{(j)}_{\mathbf g}\) in this case can be approximated as follows.
\begin{equation}
\begin{aligned}
    {\boldsymbol{\Sigma }}^{(j)}_{\mathbf g} &=\left[{\mathbf {\ddot{B}}^H} \mathbf{R}_P^{-1} {\mathbf {\ddot{B}}}+\mathbf {I}_{L} \otimes \left(\mathbf {\widehat \Gamma }^{(j-1)}\right)^{-1}\right]^{-1}\\
    &= \left[{ \mathbf {I}_{L} \otimes \mathbf {B}^H} \mathbf{R}^{-1} {\mathbf B} + \mathbf {I}_{L} \otimes \left(\mathbf {\widehat \Gamma }^{(j-1)}\right)^{-1} \right]^{-1}\\
    &\approx \mathbf {I}_{L} \otimes \left[{ \mathbf {B}^H} \mathbf{R}^{-1} {\mathbf B} + \left(\mathbf {\widehat \Gamma }^{(j-1)}\right)^{-1} \right]^{-1} = \mathbf {I}_{L} \otimes {\boldsymbol{\Sigma }}^{(j)},
\end{aligned}
\end{equation}
where ${\boldsymbol{\Sigma }}^{(j)} = \left[{ \mathbf {B}^H} \mathbf{R}^{-1} {\mathbf B} + \left(\mathbf {\widehat \Gamma }^{(j-1)}\right)^{-1} \right]^{-1}$ and $\mathbf{R} = \sigma^2_\text{AWGN}\mathbf{I}_{N^P_z}$. At high SNR, the approximation above closely approaches equality. In a similar manner, the parameter $\boldsymbol {\mu}^{(j)}_{\mathbf{g}}$ can be estimated as
\begin{equation}
    \begin{aligned}
        &\boldsymbol {\mu}^{(j)}_{\mathbf{g}} = {\mathbf{\Sigma}}^{(j)}_{\mathbf{g}} \mathbf{\ddot B}^{H} {\mathbf {R}}_{P}^{-1} {\mathbf {\overline y}}^P_{\mathbf{g}} 
        = {\mathbf{\Sigma}}^{(j)}_{\mathbf{g}} \left[ \mathbf{B}^{H} {\mathbf {R}}^{-1} \otimes \mathbf {I}_{L} \right] {\mathbf {\overline y}}^P_{\mathbf{g}} \\
        & = \left[ {\mathbf{\Sigma}}^{(j)} \mathbf{B}^{H} {\mathbf {R}}^{-1} \otimes \mathbf {I}_{L} \right] {\mathbf {\overline y}}^P_{\mathbf{g}} 
        = \text{vec}\left[ {\mathbf{\Sigma}}^{(j)} \mathbf{B}^{H} {\mathbf {R}}^{-1} \mathbf{Y} \right] = \text{vec} \left[ \mathbf{\widehat{H}}\right],
    \end{aligned}
\end{equation}
where ${\mathbf {\overline y}}^P_{\mathbf{g}} = \text{vec}(\mathbf{Y})$, $\mathbf{Y} = \left[ \mathbf{\overline{y}}_{P,1}~\mathbf{\overline{y}}_{P,2}, \cdots, \mathbf{\overline{y}}_{P,L} \right] \in \mathbb{C}^{N^P_z \times L} $, and $\mathbf{\widehat{H}} = {\mathbf{\Sigma}}^{(j)} \mathbf{B}^{H} {\mathbf {R}}^{-1} \mathbf{Y}$. Thus, the LCGBL-assisted group sparse VLC channel is derived from the \textit{a posteriori} mean, i.e., $\mathbf{\widehat h}_{\mathbf{g},\text{LCGBL}} =\boldsymbol{\mu}^{(j)}_{\mathbf{g}}$ and the corresponding error covariance matrix $\mathbf{\Sigma}^{(j)}_{\mathbf{g}}$ represents the resultant estimation uncertainty. The estimated CTF is denoted as $\mathbf{\overline{h}}_{\mathbf{g},\text{LCGBL}} = \mathbf{\ddot Q}\mathbf{\widehat h}_{\mathbf{g},\text{LCGBL}}$. In this case, the estimated CFR of each frame is $\mathbf{\overline{h}}_{i,\text{LCGBL}} = \bigg \{ \mathbf{\overline{h}}_{\mathbf{g,\text{LCGBL}}} \left[ (i-1)L + n \right] \bigg \}^{M}_{n=1}$.\\ 

\subsection{Computational Complexity Analysis}
\textcolor{black}{In this subsection, we analyze the computational complexity of the proposed CE methods. The complexity of the GBL-based sparse CE scheme is on the order of $\mathcal{O}(L_h^{3}L^{3})$, which is dominated by an inversion of an $L_hL \times L_hL$ matrix. Likewise, the LCGBL-based sparse CE scheme incurs a complexity of $\mathcal{O}(L_h^{3})$ due to the inversion of an $L_h \times L_h$ matrix. In contrast, the worst-case complexity of the GOMP algorithm scales as $\mathcal{O}\big(L^{3}(N_z^{P})^{3}\big)$, since an intermediate LS estimate must be computed at each iteration. As also demonstrated in Section V, LCGBL attains NMSE performance close to that of GBL. Therefore, the proposed LCGBL method provides an attractive tradeoff, offering both competitive performance and reduced computational burden. For completeness, the conventional LMMSE-based CE approach also has complexity on the order of $\mathcal{O}(L_h^{3})$. The comprehensive derivations of the computational and space complexities associated with various schemes have been deferred to our technical report \cite{shubPAPR}.}

\begin{figure*}[t]
	\centering
        \captionsetup[subfigure]{justification=centering}
	\subfloat[]{\label{sh50}\includegraphics[width=45mm,height=44mm]{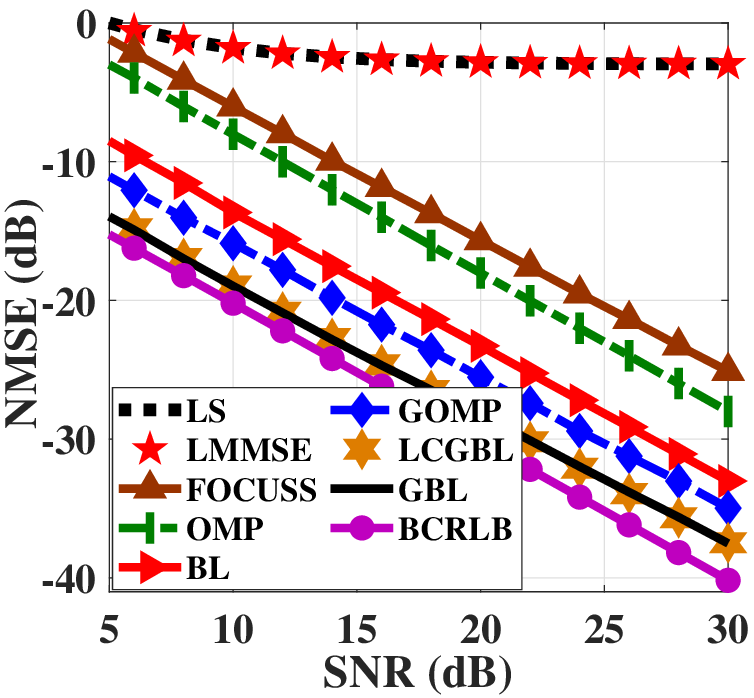}}
	\subfloat[]{\label{r12}\includegraphics[width=45mm,height=44mm]{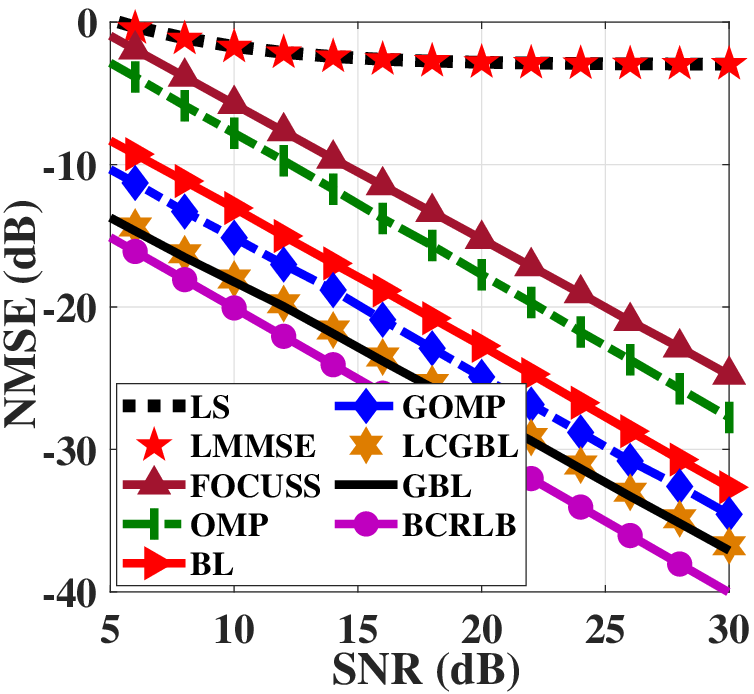}}
        \subfloat[]{\label{sh51}\includegraphics[width=45mm,height=44mm]{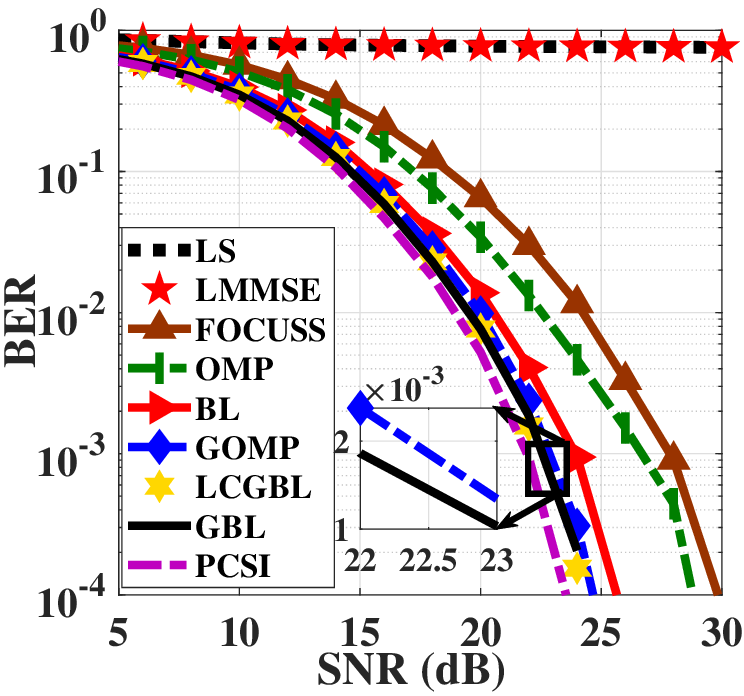}}
	\subfloat[]{\label{sh52}\includegraphics[width=45mm,height=44mm]{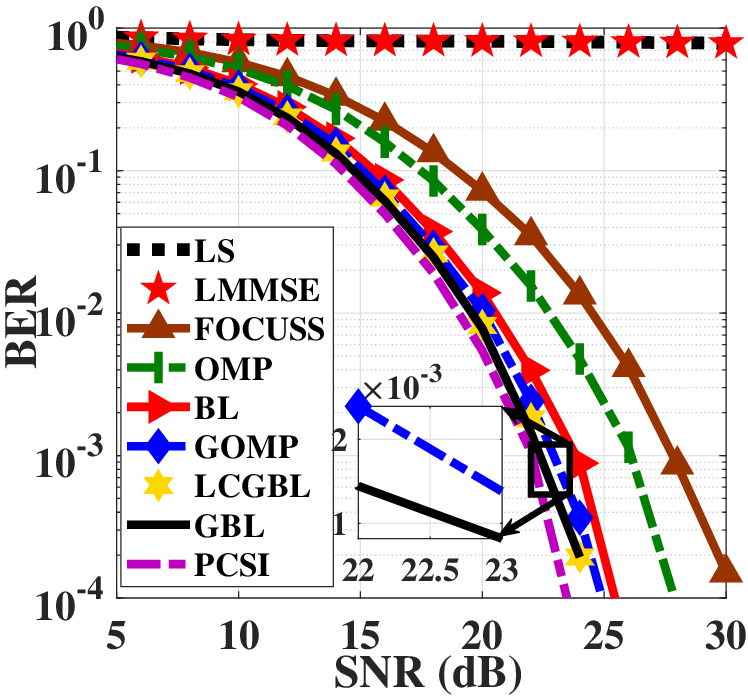}}
	\caption{MMV-based reduced-PAPR VLC system, utilizing $16$ QAM modulation, $L_h = 40$ and $L = 6$, (a) NMSE versus SNR performance in the DCO-OFDM system; (b) NMSE versus SNR performance in the ACO-OFDM system; \textcolor{black}{(c) BER versus SNR performance in the DCO-OFDM system; (d) BER versus SNR performance in the ACO-OFDM system.}}
	\label{Ab1}
\end{figure*}
\section{Simulation Results}
This section portrays our simulation results illustrating and comparing the efficacy of the proposed GBL-assisted technique, in comparison to that of GOMP \cite{srivastava2021bayesian} and SMV-based sparse SR schemes such as BL, OMP, and FOCUSS, \cite{shub,srivastava2021sparse}, as well as to the conventional LMMSE and LS schemes \cite{zhao2013channel}. Explicitly, we quantify the NMSE, BER, and OP for the CE-aided SR of the simultaneous group-sparse multipath O-OFDM-VLC CIR and FD O-OFDM symbols. The BER reflects the accuracy of detection achieved by the receivers relying on the CSI estimates, while the OP measures the likelihood that the highest attainable data rate is beneath the threshold $C_{\text{\textzeta}}$, i.e., we have $Pr(\log_2(\text{SNR}+1) \leq C_{\text{\textzeta}}) = \text{\textzeta}$ \cite{5226964}. As part of our simulations, \textcolor{black}{$\mathbf{\ddot{R}}_{h} = \left ({\mathbf {I}_{L} \otimes \mathbf{I}_{L_h}}\right)$, $\mathbf{R}_P = \sigma_{\text{AWGN}}^2 \mathbf{I}_{LN_z^P}$}, and the SNR in decibels (dB) is given as \textcolor{black}{SNR (dB)} = $10\log_{10}$\(\left(\displaystyle \frac{P_{\overline{x}}}{\sigma_{\text{AWGN}}^2}\right)\). Lastly, the NMSE is characterized by $\text{NMSE} = \frac{||\mathbf{\widehat{h}_g}-\mathbf{h_g}||^2_2}{||\mathbf{h_g} ||^2_2}$. The GBL's and BL's stopping criteria are as follows, $\epsilon = 10^{-6}\text{ and}\ m_{\max} = 50$, whereas that of the GOMP/OMP is set to $\text{\textctc} = 0.1$. In the context of FOCUSS, the noise variance is used as the regularization parameter, the stopping threshold is $10^{-5}$ with the $l_p$-norm parameter initialized to $p = 0.8$. Additionally, the maximum number of iterations is $800$. 
The system parameters are summarized in Table \ref{table2}, unless stated otherwise. 
\begin{table}[t]
\centering
\caption{\bf O-OFDM-VLC system design parameters}
\label{table2}
\resizebox{!}{!}{%
\begin{tabular}{|l|c|}\hline
\textbf{Parameter} & \textbf{Value} \\ \hline 
Room dimension & $5$ m $\times$ $5$ m $\times$ $3$ m \\
Modulation & DCO-OFDM, ACO-OFDM \\
DC bias & $7$ dB \\
Over-sampling factor & $4$ \\
No. of pilot subcarriers ($N_z^P$) & $32$ \\
No. of data subcarriers ($N_z^D$) & $32$ \\
CIR order (L$_h$) & $40$  \\ 
Outage capacity $(C_{\text{\textzeta}})$ & $5$ bps/Hz \\ 
No. of dominant paths & $6$   \\
No. of frames ($L$) & $6$ \\
Compression factor ($\alpha$) & $31/32$ \\ \hline
\end{tabular}%
}
\end{table}

Fig. \ref{Ab1}\subref{sh50} contrasts the NMSE performance of our GBL-based technique to alternative approaches in the context of our DCO-OFDM system. The GBL approach achieves significantly lower NMSE than the GOMP, standard SMV-based BL, OMP, FOCUSS, LMMSE, and LS methods. This improvement stems from the GBL's capability of exploiting the concatenated multipath VLC CIR's simultaneous group-sparsity across all the measurement vectors, overcoming the limitations of the OMP and BL. The GOMP and OMP methods exhibit inferior performance due to their sensitivity to the stopping parameter (set to $\text{\textctc} = 0.1$) and dictionary matrix. However, the proposed LCGBL scheme demonstrates comparable efficiency to GBL, as its NMSE performance is nearly identical to that of GBL, despite its substantially reduced computational cost. Performance comparisons with the MMV-based FOCUSS also show lower NMSE than the GBL due to user-defined parameters and convergence issues. The conventional LS and LMMSE methods exhibit lower performance due to their inability to exploit the inherent sparsity in the VLC system's multipath CIR. By contrast, the GBL-assisted sparse estimation successfully leverages sparsity through EM-assisted hyperparameter estimation, ensuring robust convergence without requiring any tuning or regularization of the parameters. Additionally, the dictionary matrix selection does not affect GBL's performance. Furthermore, NMSE comparisons with the appropriately normalized BCRLB exhibit similarity for the NMSE achieved by the GBL and LCGBL to the BCRLB, confirming the effectiveness of the GBL and LCGBL methods.
\begin{figure*}[t]
	\centering
        \captionsetup[subfigure]{justification=centering}
	\subfloat[]{\label{sh53}\includegraphics[width=45mm,height=44mm]{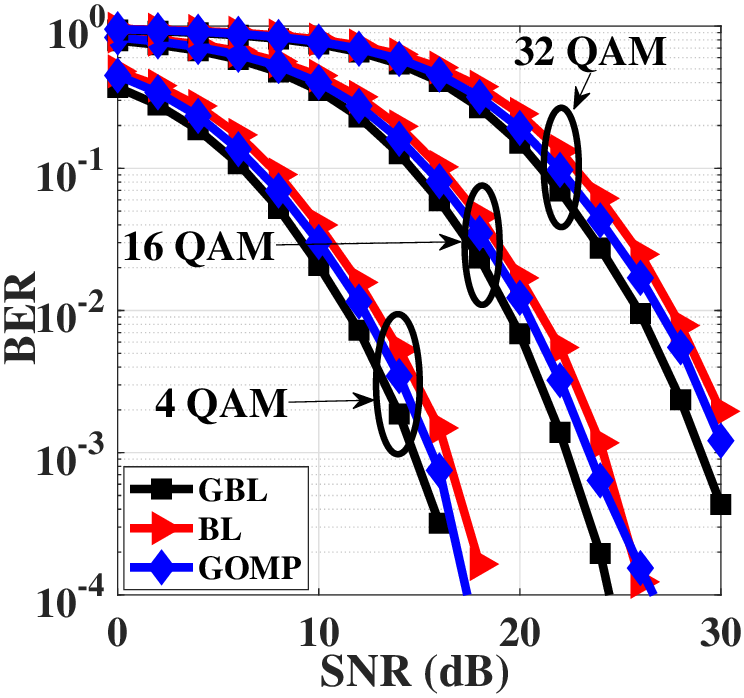}}
	\subfloat[]{\label{sh54}\includegraphics[width=45mm,height=44mm]{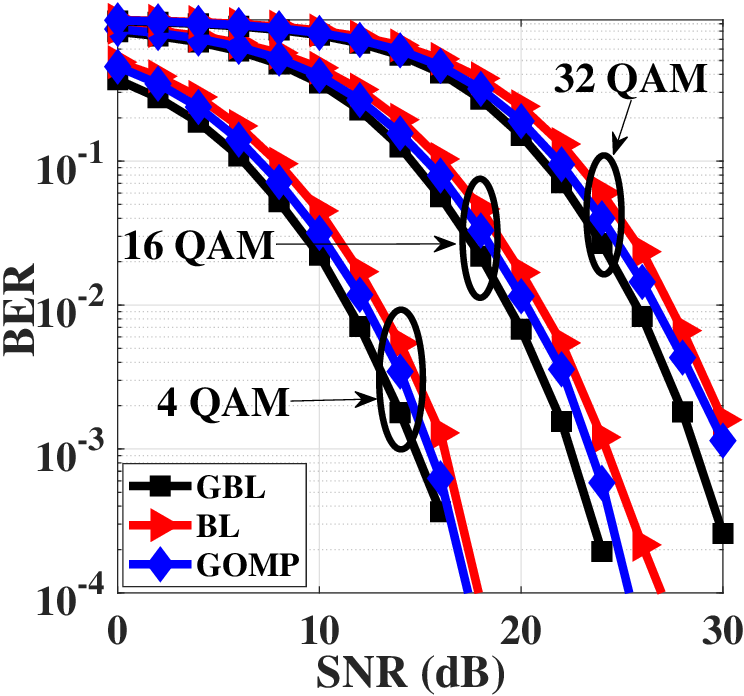}}
        \subfloat[]{\label{sh55}\includegraphics[width=45mm,height=44mm]{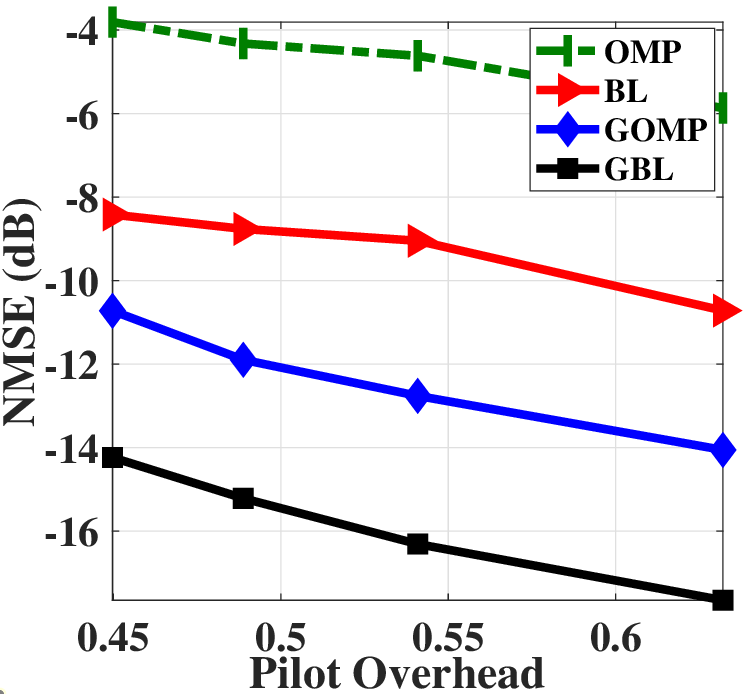}}
	\subfloat[]{\label{sh56}\includegraphics[width=45mm,height=44mm]{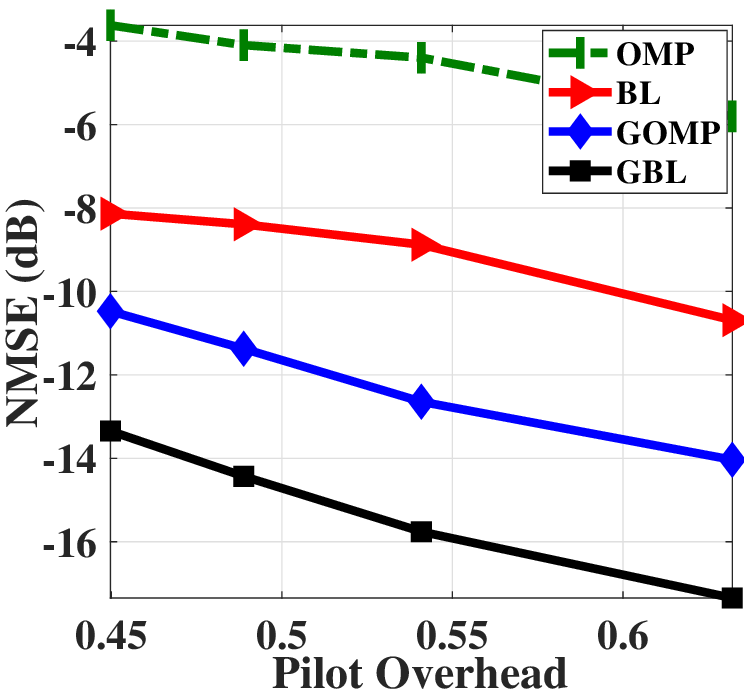}}
	\caption{MMV-based reduced-PAPR VLC system utilizing $L_h = 40$ and $L = 6$, (a) BER versus SNR performance in the DCO-OFDM system; (b) BER versus SNR performance in the ACO-OFDM system; (c) NMSE versus pilot overhead performance for $16$ QAM modulation and SNR $= 10$ dB in DCO-OFDM system; (d) NMSE versus pilot overhead performance for $16$ QAM modulation for SNR $= 10$ dB in the ACO-OFDM system.}
	\label{R7}
\end{figure*}
\begin{figure*}[t]
	\centering
        \captionsetup[subfigure]{justification=centering}
	\subfloat[]{\label{sh59}\includegraphics[width=45mm,height=44mm]{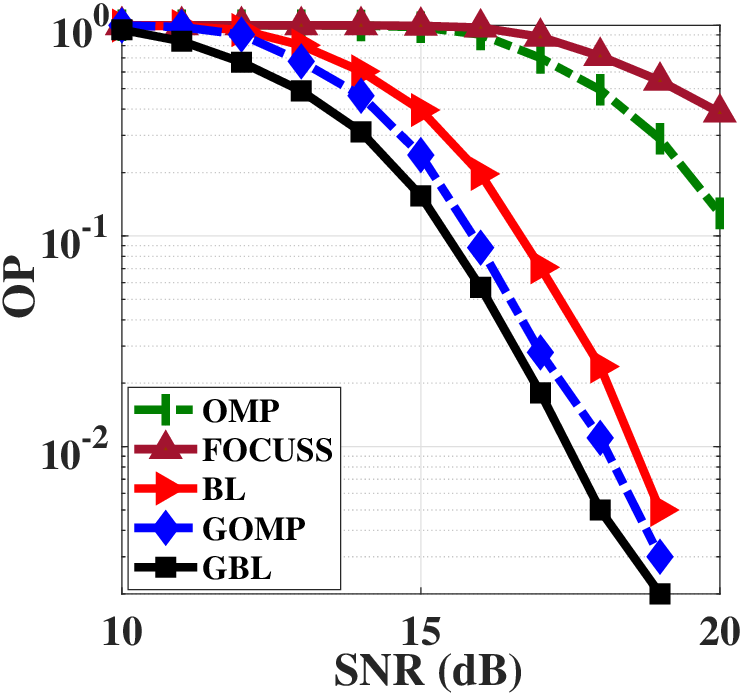}}
	\subfloat[]{\label{sh60}\includegraphics[width=45mm,height=44mm]{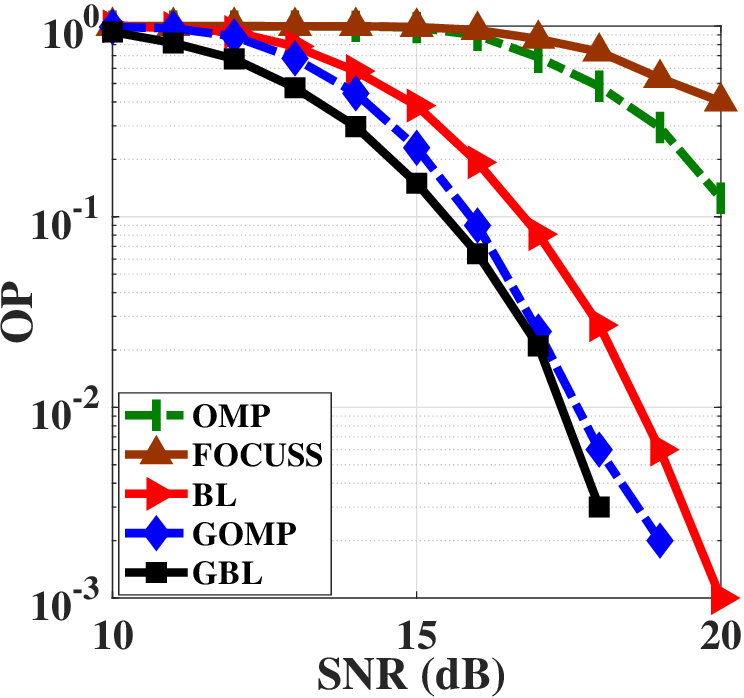}}
        \subfloat[]{\label{F1}\includegraphics[width=45mm,height=44mm]{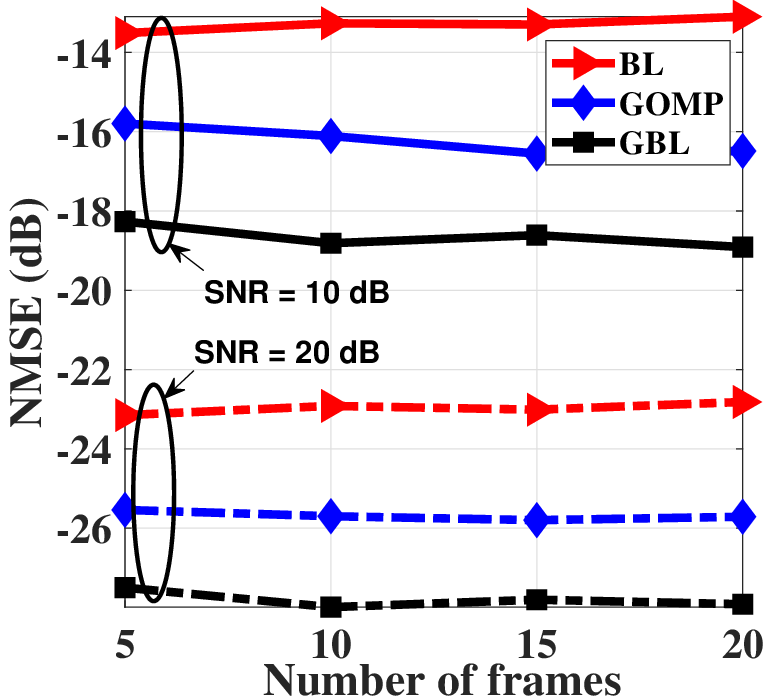}}
	\subfloat[]{\label{F2}\includegraphics[width=45mm,height=44mm]{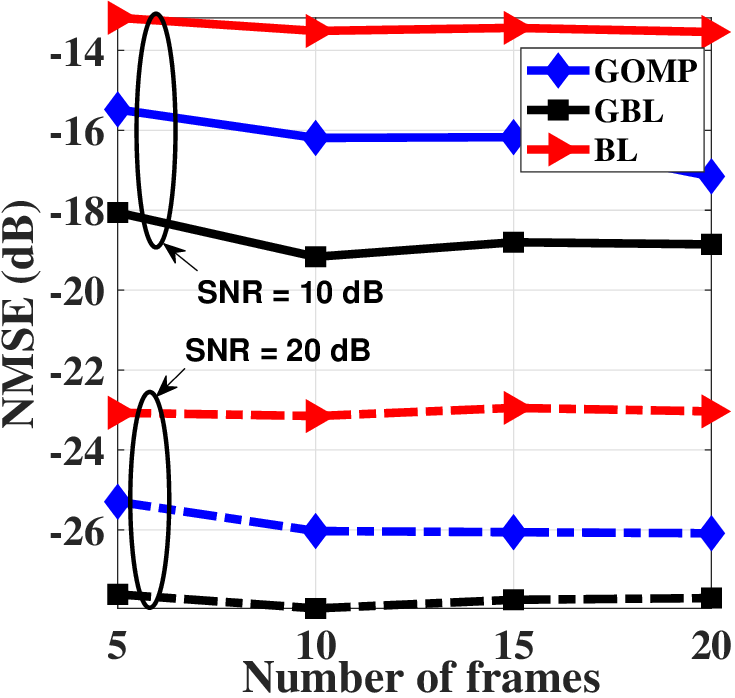}}
	\caption{MMV-based reduced-PAPR VLC system with $16$ QAM modulation and $L_h = 40$, (a) OP versus SNR performance for $L = 6$ in the DCO-OFDM system; (b) OP versus SNR performance for $L = 6$ in the ACO-OFDM system; (c) NMSE versus the number of training frames in the DCO-OFDM; (d) NMSE versus the number of training frames in the ACO-OFDM system.}
	\label{R10}
\end{figure*}
\begin{figure*}[t]
	\centering
        \captionsetup[subfigure]{justification=centering}
    \subfloat[]{\label{1mn}\includegraphics[width=50mm,height=45mm]{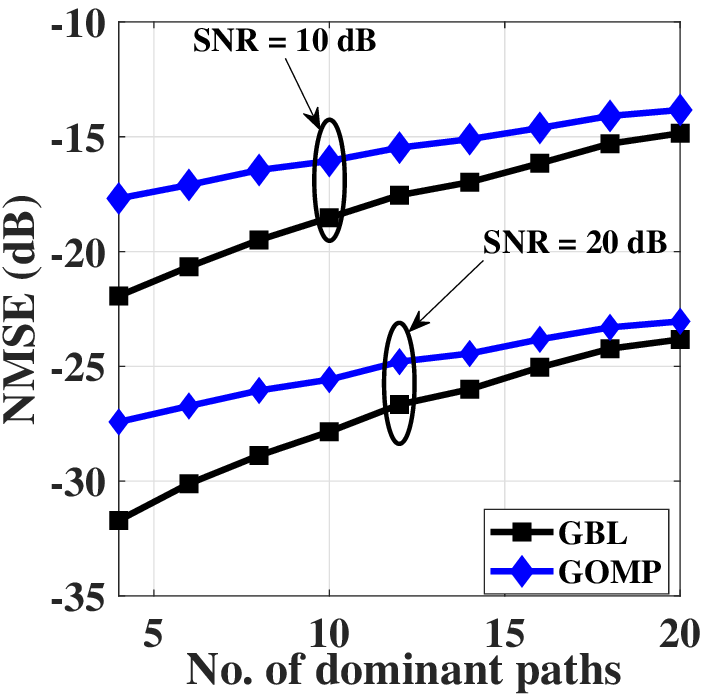}}\hspace{5mm}
	\subfloat[]{\label{rev1}\includegraphics[width=50mm,height=45mm]{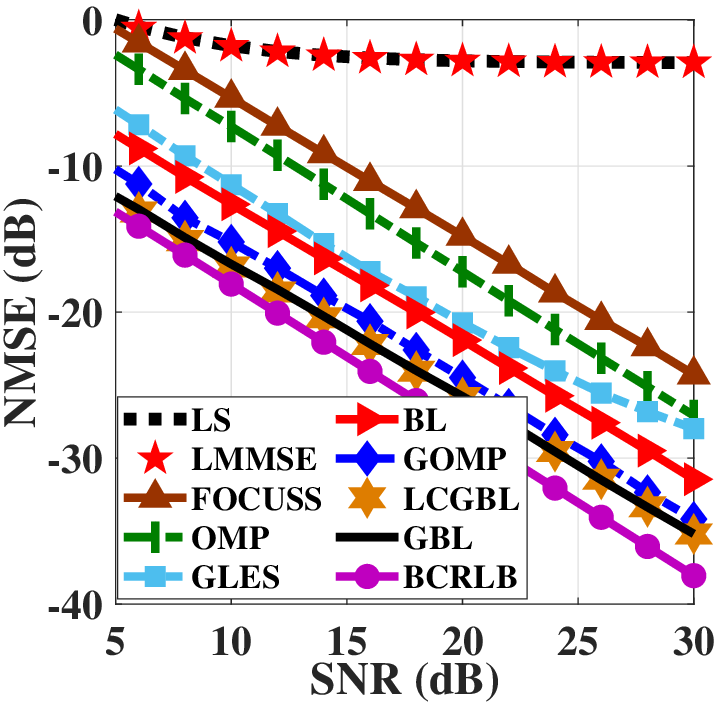}}\hspace{5mm}
	\subfloat[]{\label{rev2}\includegraphics[width=50mm,height=45mm]{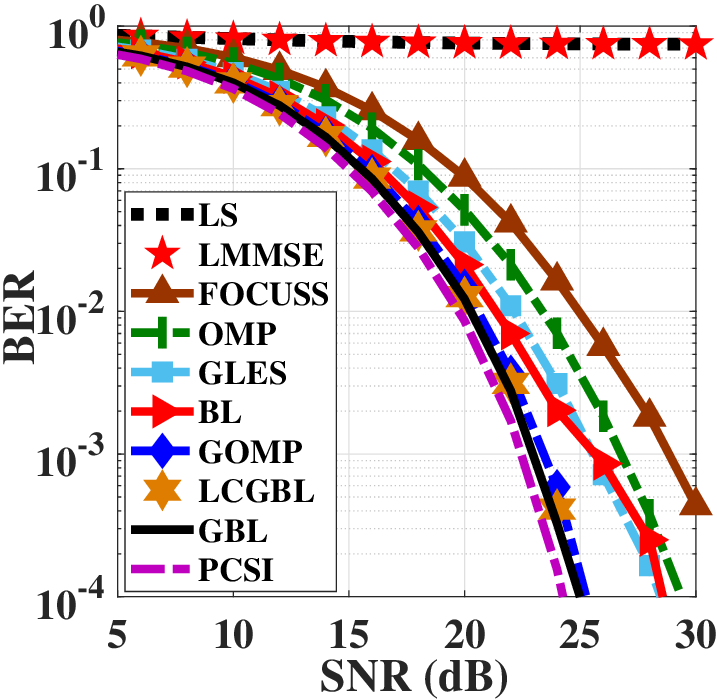}}
	\caption{\textcolor{black}{MMV-based reduced-PAPR DCO-OFDM VLC system, (a) NMSE versus number of dominant paths performance, utilizing $4$ QAM modulation, $L_h = 40$ and $L = 6$; (b) NMSE versus SNR performance, utilizing $16$ QAM modulation, $L_h = 45$ and $L = 10$; (c) BER versus SNR performance, utilizing $16$ QAM modulation, $L_h = 45$ and $L = 10$.}}
	\label{1}
\end{figure*}
Fig. \ref{Ab1}\subref{r12} depicts a comparison of the NMSE performance to that of the various estimation techniques employed in the ACO-OFDM system. The GBL and LCGBL approaches produced an overall improved result in this instance as well, following a similar trend. 


\textcolor{black}{The BER performance of the considered techniques for both DCO-OFDM and ACO-OFDM is shown in Fig. \ref{Ab1}\subref{sh51} and Fig. \ref{Ab1}\subref{sh52}, respectively. As observed, the proposed GBL-assisted CE scheme achieves a consistently lower BER than the competing techniques owing to its superior estimation accuracy. To highlight this performance difference more clearly, zoomed versions of the BER plots have also been incorporated. Moreover, the gain can be quantified in terms of the SNR reduction required to attain the same target BER. For example, in Fig. \ref{Ab1}\subref{sh51}, at a BER of $10^{-2}$, the proposed GBL-based receiver requires approximately $0.5$ dB lower SNR than the GOMP-based receiver. A similar trend is observed for LCGBL, which also achieves lower BER than GOMP while maintaining significantly reduced complexity. Furthermore, the BER achieved by the proposed GBL-based receiver remains close to that of the receiver with perfect CSI, thereby highlighting its strong capability for reliable signal recovery.}

In the DCO-OFDM system, Fig. \ref{R7}\subref{sh53} evaluates the BER efficacy of the GBL-assisted method against the GOMP and BL techniques for various QAM schemes, namely for $4, 16,$ and $32$. Regardless of the modulation order, it is clear that the proposed BL-based scheme outperforms the other methods. A comparable BER performance trend is observed for the ACO-OFDM system, as shown in \ref{R7}\subref{sh54}.

Significantly, the CSI estimation framework based on the GBL methodology transmits $LN^P_z$ pilot symbols for $L$ OFDM blocks. Consequently, the pilot overhead normalization factor is calculated as $\rho = \frac{N^P_z}{N^P_z+N^D_z}$, where $N^P_z$ takes values of $22$, $32$, $42$, and $52$. Fig. \ref{R7}\subref{sh55} and Fig. \ref{R7}\subref{sh56} illustrate the NMSE performance for several sparse CSI estimation techniques for both the DCO-OFDM and ACO-OFDM systems as a consequence of varying $\rho$. It is noticeable that as the pilot overhead $\rho$ increases, the NMSE performance improves for all the contending methods. Additionally, for the DCO-OFDM and ACO-OFDM systems, the NMSE associated with the proposed GBL method at $\rho = 0.45$ outperforms that of GOMP at $\rho = 0.55$. This validates that the GBL-based estimation approach leads to a successful reduction in $\rho$, while satisfying an established threshold of NMSE, thus highlighting its enhanced bandwidth efficiency.

In Fig. \ref{R10}\subref{sh59}, we illustrate a comparative analysis of the GBL with the other competing approaches used in the DCO-OFDM system. The results illustrated in Fig. \ref{R10}\subref{sh59} suggest that the OP of the proposed GBL scheme is lower than that of the aforementioned approaches. This outcome can be ascribed to the superior estimation accuracy of the GBL-based method when contrasted to the other CE approaches mentioned above. Fig. \ref{R10}\subref{sh60} depicts the OP performance for the ACO-OFDM system, which reaffirms the hypothesis that the GBL-based scheme excels over the conventional CE and sparse estimation methods.

Fig. \ref{R10}\subref{F1} illustrates the NMSE performance of the proposed and existing methods for our MMV-based DCO-OFDM system. The comparison is conducted by varying the number of training frames $L$ from $5$ to $20$. It is noteworthy that the GBL scheme with as few as $L = 5$ frames achieves a performance level comparable to that of the GOMP scheme having $L=20$ training frames, demonstrating the robustness of the former. The NMSE performance associated with varying $L$ in the case of the ACO-OFDM VLC system is shown in Fig. \ref{R10}\subref{F2}, demonstrating once again that the GBL-based scheme outperforms the GOMP technique. Therefore, the GBL has the potential to reduce the training overhead in both the O-OFDM systems while maintaining a desired level of estimation accuracy. 

\begin{figure*}[h]
	\centering
        \captionsetup[subfigure]{justification=centering}
	\subfloat[]{\label{sh581}\includegraphics[width=48mm,height=45mm]{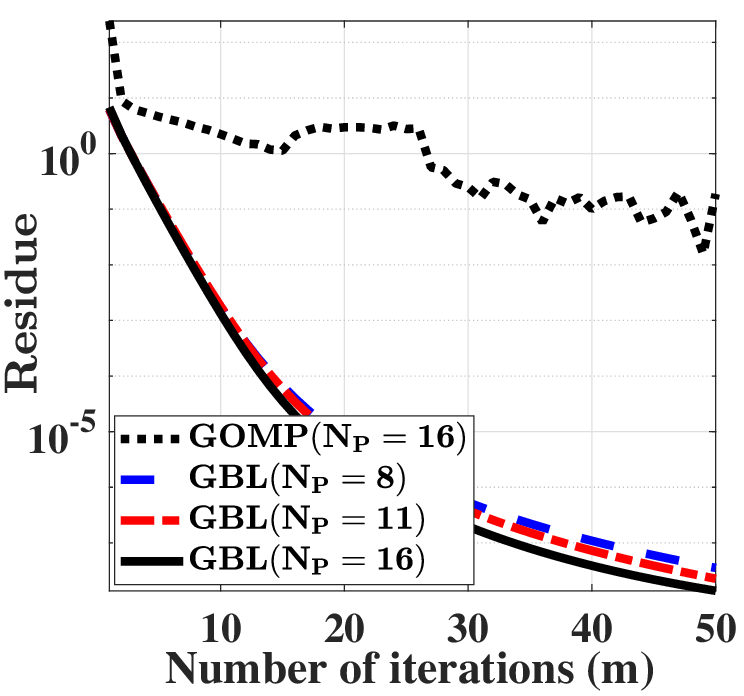}}\hspace{3mm}
    \captionsetup[subfigure]{justification=centering}
	\subfloat[]{\label{sh5811}\includegraphics[width=60mm,height=48mm]{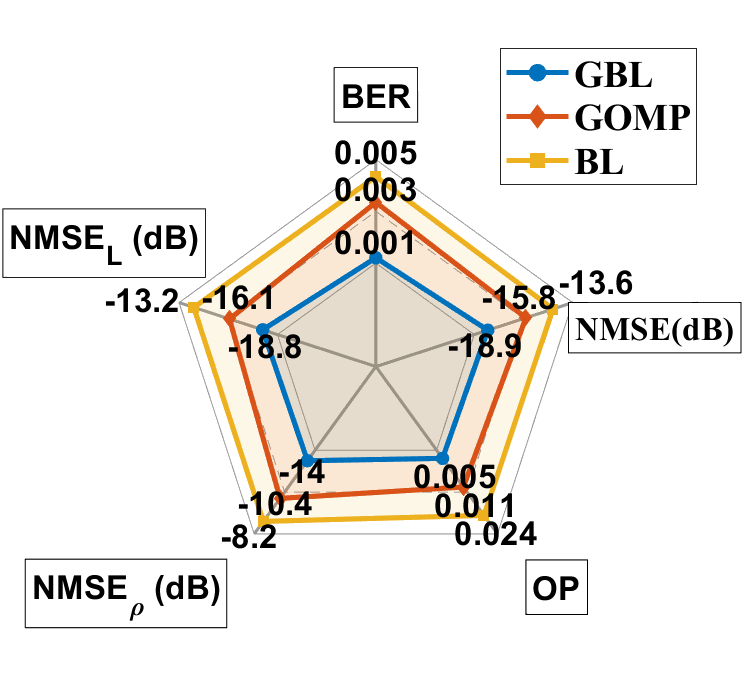}}\hspace{3mm}
 	\subfloat[]{\label{sh582}\includegraphics[width=60mm,height=48mm]{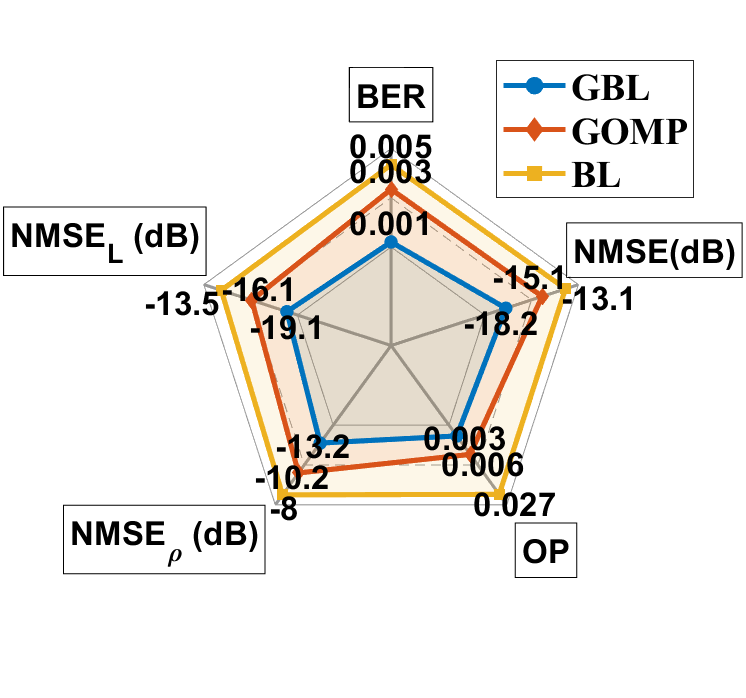}}
	\caption{(a) Residue versus number of iterations (m) for $L_h = 40,~\text{SNR}=5~\text{dB}$ and $16$ QAM modulation; (b) Stylized performance comparison of the reduced-PAPR VLC system for DCO-OFDM; (c) Stylized performance comparison of the reduced-PAPR VLC system for ACO-OFDM.}
	\label{conv}
\end{figure*}

\textcolor{black}{Fig. \ref{1}\subref{1mn} presents the NMSE performance of the MMV-based reduced-PAPR DCO-OFDM VLC system under varying sparsity conditions, characterized by different numbers of dominant paths in the VLC CIR, ranging from $4$ to $20$, for the proposed GBL- and GOMP-based CE techniques. As the number of dominant paths increases, both methods exhibit a degradation in NMSE performance owing to the reduction in channel sparsity. Nevertheless, the proposed GBL-based method consistently outperforms the GOMP-based counterpart across all considered sparsity levels. Notably, the GBL approach achieves an NMSE comparable to that of the GOMP technique with only $4$ dominant paths, even when the number of dominant paths is increased to $20$. This demonstrates the stronger robustness of the proposed GBL method under a broad range of VLC channel conditions.}

\textcolor{black}{To further demonstrate the advantages of the proposed method, we evaluated the advocated GBL-based approach under different simulation settings, where the number of dominant paths is $8$, the modulation scheme is $16$-QAM, $L=10$, and $L_h=45$. Figure \ref{1}\subref{rev1} compares the NMSE performance of the considered techniques, including the group LASSO with effective support (GLES) method, for the DCO-OFDM VLC system. The regularization parameter of the GLES scheme was empirically tuned to obtain its best NMSE performance. Since the performance of GLES is highly sensitive to the selection of this user-defined regularization parameter, its NMSE remains inferior to that achieved by the proposed GBL-based method. It is therefore evident that the proposed GBL- and LCGBL-based approaches outperform the other considered techniques. In addition, Fig. \ref{1}\subref{rev2} presents the BER performance. It can be observed that the BER attained by the proposed GBL- and LCGBL-based CE techniques approaches that of the receiver with PCSI, while remaining lower than that of the GOMP- and BL-based counterparts due to their improved estimation accuracy.}


Fig. \ref{conv}\subref{sh581} depicts the convergence of our proposed GBL and GOMP-based technique as a function of the number of EM iterations for various values of the number of pilot symbols $N_P$. This represents the number of EM iterations required for the convergence of the hyperparameters. For a fixed value of the convergence parameter $\epsilon = 10^{-6}$, it is evident that the number of iterations required for convergence decreases as $N_P$ increases. Furthermore, one can conclude that our GBL-based technique converges faster than the non-Bayesian technique, such as GOMP.

Finally, a comprehensive performance visualization of the three techniques, namely GBL, GOMP, and conventional BL, with respect to the DCO-OFDM, has been presented in Fig. \ref{conv}\subref{sh5811}, with regard to several metrics such as the BER, NMSE, pilot overhead (NMSE$_{\rho}$), number of frames (NMSE$_{L}$), and OP for the VLC system.
In this investigation, the OP has been examined at an SNR of $18$ dB, NMSE and BER have been evaluated at an SNR of $10$ dB and $22$ dB, NMSE$_{\rho}$ has been assessed at $\rho = 0.45$, and NMSE$_{L}$ is computed at $L = 10$ and SNR $= 10$ dB. According to the results, our GBL strategy offers the optimum performance in terms of the OP, BER, and NMSE metrics for a particular SNR. Additionally, the GBL-based technique produces the best NMSE for the chosen values of parameters of $\rho$ and $L$. As depicted in Fig. \ref{conv}\subref{sh5811}, the proposed GBL approach occupies the smallest area within the polygon. This indicates that it offers the lowest NMSE, OP, BER, NMSE$_{L}$, and NMSE$_{\rho}$ values, making it the most promising choice for the estimation of CIR in DCO-OFDM-based VLC systems. A comparable in-depth evaluation of the performance is undertaken regarding the ACO-OFDM modulation in Fig. \ref{conv}\subref{sh582}. The thresholds are identical to those of Fig. \ref{conv}\subref{sh5811}, and the trends are also similar. As stated previously, the GBL-based technique possesses the smallest region of the polygon, indicating its preeminent performance. Consequently, for the VLC system, the GBL-assisted sparse recovery method advocated is perfectly suited for the O-OFDM system, since it enhances the quality of CSI estimation, resulting in overall superior performance in terms of the NMSE, BER, NMSE$_{L}$, and OP metrics, together with a low pilot overhead.
\section{Summary and Conclusion}
A novel GBL-based technique has been conceived for the estimation of group-sparse multipath CIR and SR of the FD O-OFDM symbols in an IM/DD-based VLC system. The methods proposed for CIR and SR are based on the MMV-based sparse wideband VLC CIR model that captures the NLoS and LoS paths, and a reduced-PAPR sparse FD O-OFDM symbol model, respectively. Our technique exploited the sparsity, which resulted in benefits in comparison to the conventional CE methods, such as LMMSE and LS. While GOMP techniques were also presented for sparse CIR and SR, they suffered from sensitivity to the choice of the dictionary matrix and error propagation. The GBL also led to faster MMV-based multipath VLC CIR estimation, addressing the slow convergence of existing algorithms. This improvement arises from the GBL’s ability to exploit the concatenated multipath VLC CIR’s simultaneous group-sparsity across all measurement vectors. Furthermore, a low-complexity variant of GBL, termed LCGBL, has been developed to significantly reduce its computational cost. The proposed method was then harnessed in ACO-OFDM and DCO-OFDM systems, and it led to a substantial reduction in the pilot overhead. Our results clearly evidenced the superior performance of GBL over the existing estimation methods, such as BL, GOMP, OMP, FOCUSS, LS, and LMMSE across a comprehensive set of performance metrics, including the OP, MSE, pilot overhead, and BER. The NMSE of the GBL was also shown to closely match the BCRLB performance bound, highlighting the practical advantages inherent in deploying the algorithms devised. \textcolor{black}{A systematic comparison with recent DL-based sparse estimation methods for VLC is an interesting and meaningful direction for future work.}


\vspace{-1mm}
\section*{Appendix: Proof of Theorem $1$}
The expectation-step (E-step) assesses $\mathbf {\mathcal {L}}\left ({\boldsymbol{\Gamma }| \widehat {\boldsymbol{\Gamma }}^{(j-1)}}\right)$ as
\begin{align} \label{App1}
\mathbb {E} \bigg \{ \log \left [{ f\left ({ {\mathbf {h_g}}; \boldsymbol{\Gamma } }\right) }\right] \bigg \} + \mathbb {E} \bigg \{ \log \left [{ f\left ( {\mathbf {\overline{y}}}^P_{\mathbf{g}} | \mathbf{h_g} \right) }\right] \bigg \}.
\end{align}
It is inferred from (\ref{App1}) that the log-likelihood term $\log \left [{ f\left ( {\mathbf {\overline{y}}}^P_{\mathbf{g}} | \mathbf{h_g} \right) }\right]$ is independent of $\boldsymbol{\Gamma }$. Therefore, it is permissible to disregard the initial term of (\ref{App1}) during the subsequent maximization step (M-step). Consequently, the M-step maximizes $\mathbf {\mathcal {L}}\left ({\boldsymbol{\Gamma }| \widehat {\boldsymbol{\Gamma }}^{(j-1)}}\right)$ relative to $\boldsymbol \Gamma$ as
\begin{align} \label{App2}
\widehat {\boldsymbol \Gamma }^{(j)}&=\arg \max _{\boldsymbol \Gamma } \mathbb {E} \bigg \{ \log \left [{ f\left ({ {\mathbf {h_g}}; \boldsymbol{\Gamma } }\right) }\right] \bigg \} \\
&=\arg \max _{\boldsymbol \Gamma } \sum \limits _{k=0}^{L_h-1} \left[{-L\log (\gamma _{k}) - \sum _{l=1}^{L} \displaystyle \frac {\mathbb {E} \big \{ \vert \mathbf {h}_{l}(k) \vert ^{2} \big \}}{\gamma _{k}} }\right].
\end{align}
The hyperparameter estimates $\widehat {\gamma }_{k}^{(j)}$ can be derived by differentiating (\ref{App2}) with respect to ${\gamma }_{k}$ and equating it to zero yielding:
\begin{equation} \label{App5}
\widehat {\gamma }_{k}^{(j)}= \frac {1}{L} \sum _{l=1}^{L} \mathbb {E}_{  {\mathbf {h_g}} \vert {\mathbf {\overline y}}^P_{\mathbf{g}}; \widehat {\boldsymbol{\Gamma }}^{(j-1)}} \bigg \{ \vert \mathbf {h}_{l}(k) \vert ^{2} \bigg \}. 
\end{equation}
To determine the conditional expectation $\mathbb {E}_{{{\mathbf {h_g}} | {\mathbf {\overline y}}^P_{\mathbf{g}}; \widehat {\boldsymbol{\Gamma }}^{(j-1)}}} \big \{ {\cdot } \big \}$, we utilize the \textit{a posteriori} pdf $f\biggl( {\mathbf {h_g}} \vert {\mathbf {\overline y}}^P_{\mathbf{g}}; \widehat {\boldsymbol{\Gamma }}^{(j-1)} \biggl)$ of $\mathbf{h_g}$, that is formulated as 
\begin{equation}\label{App4}
f\left( {\mathbf {h_g}} \vert {\mathbf {\overline y}}^P_{\mathbf{g}}; \widehat {\boldsymbol{\Gamma }}^{(j-1)} \right) = \mathcal {CN}\left (\boldsymbol{\mu}^{(j)}_{\mathbf{g}}, {\boldsymbol{\Sigma }}_{\mathbf{g}}^{(j)}\right).
\end{equation}
By substituting $\boldsymbol{\Gamma } = \widehat {\boldsymbol{\Gamma }}^{(j-1)}$ into the MMSE error covariance matrix $\boldsymbol{\Sigma}_{\mathbf{g}}$ of (\ref{App3}), one can obtain the \textit{a posteriori} covariance matrix $\boldsymbol{\Sigma}^{(j)}_{\mathbf{g}} \in \mathbb{C}^{LL_h \times LL_h}$ in (\ref{App4}). Similarly, the \textit{a posteriori} mean $\boldsymbol{\mu}^{(j)}_{\mathbf{g}}\in \mathbb{C}^{LL_h \times 1}$ of (\ref{App4}) can be obtained by setting $\boldsymbol{\Sigma}_{\mathbf{g}} = \boldsymbol{\Sigma}^{(j)}_{\mathbf{g}}$ in the expression of the estimate of the mean $\boldsymbol{\mu}_{\mathbf{g}}$ in (\ref{App3}). Thus, the hyperparameter update $\widehat {\gamma }_{k}^{(j)}$ of (\ref{App6}) can be obtained by using the conditional expectation in (\ref{App5}), which is given as
\begin{equation}
\begin{aligned}
\mathbb {E} \bigg \{ \vert \mathbf {h}_{l}(k) \vert ^{2} \bigg \} &= \boldsymbol{\Sigma}^{(j)}_{\mathbf{g}} \big [k+(l-1)L,k+(l-1)L \big] \\ 
&+\, \left \vert{ \boldsymbol {\mu}^{(j)}_{\mathbf{g}}\big [k+(l-1)L \big] }\right \vert^{2}.
\end{aligned}
\end{equation} 


\bibliographystyle{IEEEtran} 
\bibliography{citation}


 

\begin{IEEEbiography}[{\includegraphics[width=1in,height=1.25in,clip,keepaspectratio]{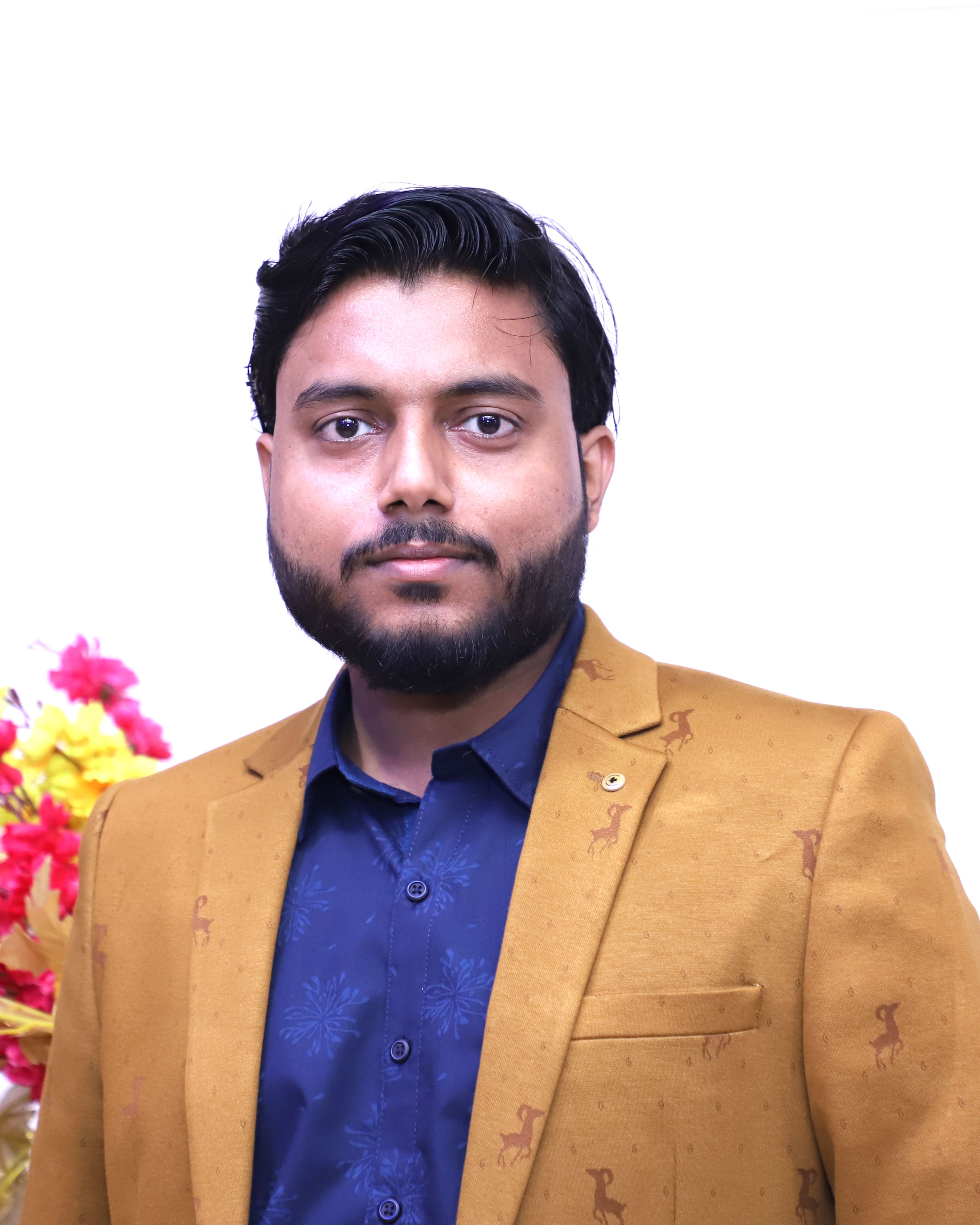}}]{Shubham Saxena}~(Graduate Student Member, IEEE) received the B.Tech. degree in Electronics and Communication Engineering from Dr. A.P.J. Abdul Kalam Technical University, Uttar Pradesh, India, in $2017$, and the M.Tech. degree in Communication System Engineering from the Indian Institute of Technology (IIT) Patna, Patna, India, in $2020$. He is currently pursuing the Ph.D. degree in Electrical Engineering at IIT Kanpur, India, where he has also been awarded the Fellowship for Academic and Research Excellence (FARE) to undertake postdoctoral-level research. His research interests include sparse Bayesian learning, channel estimation, and signal processing for visible light communication (VLC), optical OFDM/OTFS, and visible light integrated sensing and communication (VL-ISAC) in $6$G wireless systems. He received the IEEE SPS UP Chapter Graduate Student Award for $2024$, the IEEE ComSoc Student Travel Grant for IEEE ANTS $2024$, the ANRF International Travel Support for IEEE WCNC $2026$, and the IIT Kanpur Gymkhana Leadership Award for $2024-2025$. He was a joint winner in the Postgraduate category of the MVC All India Student Paper Contest at IEEE INDICON $2019$, receiving the Best Paper certificate. 
\end{IEEEbiography}
\vspace{-10mm}
\begin{IEEEbiography}[{\includegraphics[width=1in,height=1.25in,clip,keepaspectratio]{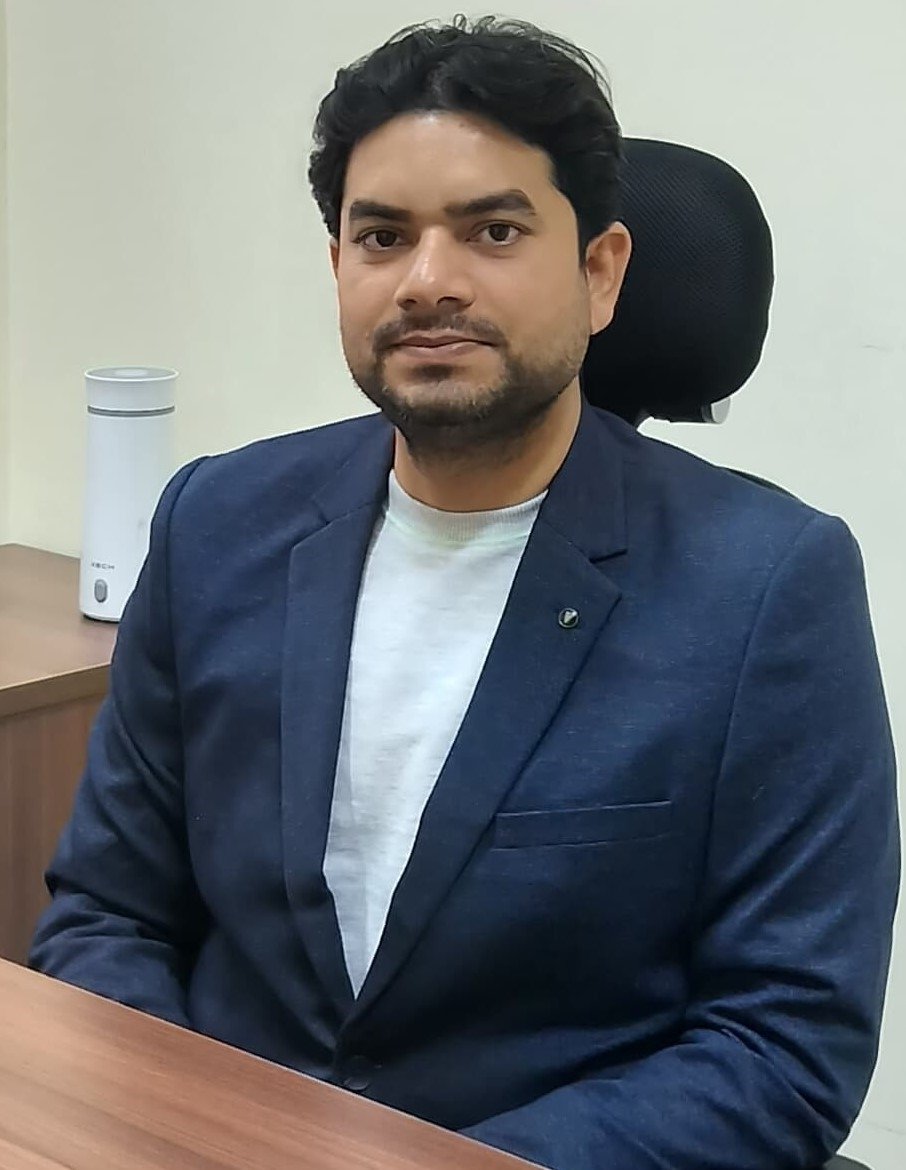}}]{Suraj Srivastava}~(MIEEE, MIET, MIEI, FIETE, YA-INAE, MNASc) received the Ph.D. degree in Electrical Engineering from Indian Institute of Technology Kanpur, India, and M.Tech. degree in Electronics and Communication Engineering from Indian Institute of Technology Roorkee, India. From October $2022$ to November $2023$, he worked as a senior lead engineer in Qualcomm India Pvt. Ltd., where he was involved in the development of $5$G-IoT modems. Previously, he was also employed as a Staff-I systems design engineer with Broadcom Research India Pvt. Ltd. and as a lead engineer with Samsung Research India, where he worked on designing Layer-$2$ of the $3$G UMTS/WCDMA/HSDPA modems. He was awarded the prestigious Qualcomm Innovation Fellowship (QIF) twice for the years $2018$ and $2022$, and received the Outstanding Ph.D. Thesis award and Silver Medal from IIT Kanpur. In $2023$ he joined Indian Institute of Technology Jodhpur, where he is currently an Assistant Professor in the Department of Electronics Engineering (ECE). He has been honoured with the prestigious INAE Young Associateship $2024$, Fellowship of IETE, IIT Jodhpur's Young Researcher Award $2025$, IIT Jodhpur's Teaching Excellence Award $2026$, IETE-N V Gadadhar Memorial Award $2026$, and Membership of NASI $2026$ in recognition to his outstanding contributions in Wireless Communications.
\end{IEEEbiography}

\begin{IEEEbiography}[{\includegraphics[width=1in,height=1.25in,clip,keepaspectratio]{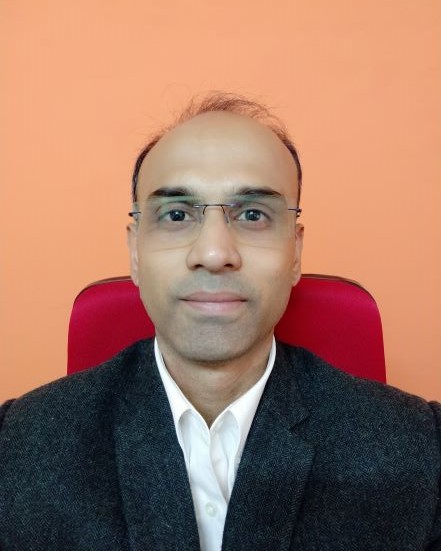}}]{Aditya K. Jagannatham}~(Senior Member, IEEE) received the bachelor’s degree from the Indian Institute of Technology, Bombay, and the M.S. and Ph.D. degrees from the University of California at San Diego, San Diego, CA, USA. From April $2007$ to May $2009$, he was employed as a Senior Wireless Systems Engineer with Qualcomm Inc., San Diego, where he was a part of the Qualcomm CDMA Technologies Division. He is currently a Professor with the Department of Electrical Engineering, IIT Kanpur, where he also held the Arun Kumar Chair Professorship. His research interests include next-generation wireless cellular and WiFi networks, with a special emphasis on various $5$G technologies, such as massive MIMO, mmWave MIMO, FBMC, NOMA, as well as emerging $6$G technologies, such as OTFS, IRS, THz systems, ISAC and VLC. He has been twice awarded the P. K. Kelkar Young Faculty Research Fellowship for excellence in research, received multiple Qualcomm Innovation Fellowships in $2018, 2022$ and $2024$, the IIT Kanpur Excellence in Teaching Award, IIT Kanpur Distinguished Teacher Award, the CAL(IT)$2$ Fellowship at the University of California at San Diego, the Upendra Patel Achievement Award at Qualcomm San Diego, and the Qualcomm $6$G UR India Gift.
\end{IEEEbiography}

\begin{IEEEbiography}[{\includegraphics[width=1in,height=1.25in,clip,keepaspectratio]{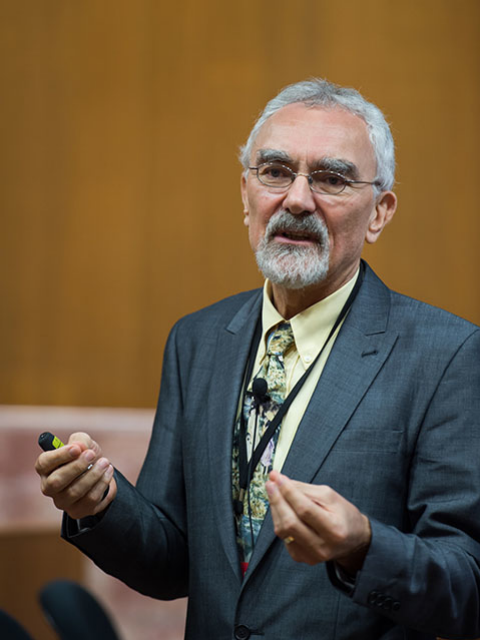}}]{Lajos Hanzo}~ (\url{http://www-mobile.ecs.soton.ac.uk},\url{ https://en.wikipedia.org/wiki/Lajos_Hanzo}) is a Fellow of the Royal Academy of Engineering, FIEEE, FIET, Fellow of EURASIP and a Foreign Member of the Hungarian Academy of Sciences. He has published widely, including $20$ Wiley-IEEE Press research monographs. He was bestowed upon the IEEE ComSoc Edwin Howard Armstrong Award and the Eric Sumner Technical Field Award.
\end{IEEEbiography}





\end{document}